\documentclass[amsthm]{elsart}

\usepackage{amsmath, amsfonts, amssymb}
\usepackage{hyperref}
\usepackage{graphicx}
\usepackage{algorithm}
\usepackage{algpseudocode}
\usepackage{xcolor}
\usepackage{arydshln, hhline}
\usepackage{algorithm}
\usepackage{algpseudocode}

\newcommand{\id}{\mathrm{id}}
\newcommand{\F}{\mathbb{F}}
\newcommand{\PS}[2]{\mathrm{P}_{#1}(#2)}
\newcommand{\PSqk}{\PS{q}{k}}
\renewcommand{\H}{\mathcal{H}}
\newcommand{\U}{\mathcal{U}}
\newcommand{\C}{\mathcal{C}}
\newcommand{\Aut}{\mathrm{Aut}}
\newcommand{\CF}{\mathrm{CF}}
\newcommand{\TR}{\mathrm{TR}}
\newcommand{\Can}{\mathrm{Can}}
\newcommand{\Gsl}{\mathrm{G^{(sl)}}}

\newcommand{\Point}[1]{\left\langle #1 \right\rangle}
\newcommand{\Fixed}[1]{\mathrm{Fixed}\left(#1\right)}
\newcommand{\FixedSeq}[1]{\overrightarrow{\mathrm{Fixed}}\left(#1\right)}
\newcommand{\Stab}[2]{\mathrm{Stab}_{{#1}}\left({#2}\right)}%
\newcommand{\Mat}[2]{\F_{q}^{#1 \times #2}}
\newcommand{\GL}[1]{\mathrm{GL}_{#1}}
\newcommand{\GammaL}[1]{\mathrm{\Gamma L}_{#1}}

\newcommand{\PGammaL}[1]{\mathrm{P\Gamma L}_{#1}}
\newcommand{\Set}[1]{\left[#1 \right]}

\newcommand{\Inn}{\mathit{Inn}}
\begin{document}
\begin{frontmatter}

\title{Canonical Forms and Automorphisms in the Projective Space}

\thanks{This research was supported by Deutsche Forschungsgemeinschaft in the priority program SPP 1489 under grant WA 1666/7-1.}

\author{Thomas Feulner}
\address{University of Bayreuth, Germany}
\ead{thomas.feulner@uni-bayreuth.de}

\begin{abstract}
Let $\C$ be a sequence of multisets of subspaces of a vector space
$\F_q^k$. We describe a practical algorithm which computes a
canonical form and the stabilizer of $\C$ under the group action of
the general semilinear group. It allows us to solve canonical form
problems in coding theory, i.e. we are able to compute canonical
forms of linear codes, $\F_{q}$-linear block codes over the alphabet
$\F_{q^s}$ and random network codes under their natural notion of
equivalence. The algorithm that we are going to develop is based on
the partition refinement method and generalizes a previous work by
the author on the computation of canonical forms of linear codes.
\end{abstract}
\end{frontmatter}

\section{Introduction}
In this paper we consider the \emph{canonization} problem defined on group actions in the following sense:
Let $G$ be a group acting on a set $X$ from the left. Furthermore, let
$\mathcal{L}(G)$ be the set of subgroups of $G$.
\begin{prob}[Canonization]
Determine a function

 \begin{align*}
\Can_G: X &\rightarrow X \times G \times \mathcal{L}(G)\\
x &\mapsto (\CF_G(x), \TR_G(x), \Stab{G}{x})
\end{align*}
with

\begin{align}
\forall x \in X, \forall g \in G: \CF_G(x) = \CF_G(gx) \\
\forall x \in X : \CF_G(x) = \TR_G(x)x \\
\forall x \in X : \Stab{G}{x} = \{g \in G \mid gx= x\}
\end{align}
The element $\CF_G(x)$ is called the \emph{canonical form} of $x$ and the element $\TR_G(x) \in G$ a \emph{transporter element}.
The element $\TR_G(x)$ is well-defined up to the multiplication with the \emph{stabilizer} $\Stab{G}{x}$ from the right.
\end{prob}

In the case of a finite group $G$ the orbit $Gx$ is finite as well. Hence, applying an orbit-stabilizer algorithm and defining $\CF(x) := \min Gx$ already solves this problem. Our goal is to define $\Can_G$ in such a way that there is an algorithm with good practical performance to compute a canonical form. Indeed,  $\Can_G$ is implicitly defined via the result of the algorithm. We included the stabilizer computation to the canonization process since the Homomorphism Principle (Theorem \ref{thm:Hom_Principle}) \cite{Laue}, which we will apply as a key tool, must have this data available.

In this work, we will provide a practical canonization algorithm for sequences of (multi-) sets of subspaces under the action of the semilinear group. It will be a natural generalization of the algorithm \cite{Feu09}, where the author solves the same problem in the special case that all occurring subspaces are one-dimensional. Since this is a question arising from coding theory, the algorithm was formulated using generator matrices of linear codes. In fact, it canonizes generator matrices of linear codes. Similarly, the present problem also has applications in coding theory as well, see Section \ref{sec:CodingTheory}.

\cite{PetrankRoth} investigate the computational complexity of the code equivalence problem for linear codes over finite fields. They show that this problem is not NP-complete. On the other hand, it is at least as hard as the graph isomorphism problem. The later problem has been studied for decades, but until now there is no polynomial time algorithm solving it. Therefore, we can not expect to give a polynomial time algorithm solving our present problem. We therefore measure efficiency in terms of running times on selected non-trivial examples.

The paper is structured as follows: The next section will describe general methods for providing practical canonization algorithms like
the \emph{partition refinement} approach. The program \emph{nauty} \cite{McKay:graphIsomorphism} is a prominent example using this idea: it canonizes a given graph under the action of the symmetric group, i.e. the relabeling of vertices.
In Section \ref{sec:reformulation}, we give a reformulation of the original problem such that we are able to use the partition refinement idea, too. The subsequent section deals with the origins of this problem from coding theory.
Subsection \ref{subsec:linCodes} summarizes the necessary modifications of Section \ref{sec:Canonization} in order to canonize linear codes. In the following, we give the details of the canonization algorithm for sequences of subspaces and finish this work with some applications of the algorithm in Section \ref{sec:Applications} and a conclusion.

\section{General canonization algorithms} \label{sec:Canonization}
This section surveys four principle attacks for the canonization of an object $x \in X$ under the action of $G$.
It is a summary of \cite{Gugisch}. Therefore, we will omit the proofs.

\subsection{The direct approach}
The first method is the most desirable. It is directly attacking the problem, which means that we understand the group action in such a way that we are able to define $\CF_G(x)$ and the group element $\TR_G(x)$ and to give a polynomial-time algorithm for its computation without making (explicitly) use of the group structure.

\begin{exmp}
Let $X$ be a totally ordered set. The symmetric group $S_n$ acts on $X^n$ via

\begin{equation*}
 \pi(x_1, \ldots, x_n) := (x_{\pi^{-1}(1)}, \ldots, x_{\pi^{-1}(n)}), \textnormal{ for all } \pi \in S_n, (x_1, \ldots, x_n) \in X^n.
\end{equation*}
We may compute a canonical form of a given sequence $(x_1, \ldots, x_n)$ by lexicographically sorting the elements of the vector. It is easy to define a transporter element. The stabilizer of $(x_1, \ldots, x_n)$ is the subgroup of elements in $S_n$ which interchanges equal entries of the vector.
\end{exmp}

\begin{exmp}
Let $\mathbb{F}_q$ be the finite field with $q$ elements, where
$q=p^r$ for some prime $p$. The \emph{general linear group}
$\GL{k}(q)$ is the set of all invertible $k\times k$-matrices with
entries in $\F_{q}$. It acts on the set $\Mat{k}{r}$ of all $k\times
r$-matrices using the usual matrix multiplication from the left. We
may define a canonical form for $M \in \Mat{k}{r}$ using the reduced
row echelon form $\operatorname{RREF}(M)$ of $M$. Gaussian
elimination is a polynomial-time algorithm to compute the canonical
form and a transporter element under this action. Furthermore, the
stabilizer of $M$ could be easily given using the stabilizer of the
canonical form $\operatorname{RREF}(M)$.

Similarly, a canonical form of $M \in \Mat{k}{r}$ under the action
of $\GL{r}(q)$ --  given by $(A, M) \mapsto MA^T$ -- could be
defined to be the one in reduced column echelon form
$\operatorname{RCEF}(M)$.
\end{exmp}

\subsection{Homomorphism Principle} \label{subsec:hom_principle}

\begin{defn}
Let $G$ be a group acting on a set $X$ and $H$ another group acting on $Y$.
A pair of mappings $(\theta:X \rightarrow Y, \varphi: G \rightarrow H)$ where $\varphi$ is a group homomorphism is called a \emph{homomorphism of group actions} if the mappings commute with the actions, i.e. $\theta(gx) = \varphi(g)\theta(x) ,\ \forall\ g \in G, x \in X$.

In the case that $\varphi$ is the identity on $G$, i.e. $G=H$, we call the function $\theta$ a \emph{$G$-homomorphism}. If the action on the right is trivial, i.e. $hy=y$ for all $y \in Y$ and $h\in H$, we call the function $\theta$ \emph{$G$-invariant}. In this case we could always suppose that $\vert H\vert =1$.
\end{defn}

\begin{thm}[Homomorphism Principle, \cite{Laue}]\label{thm:Hom_Principle}
Let $(\theta : X \rightarrow Y, \varphi: G \rightarrow H)$ be a homomorphism of group actions, with surjective mappings $\theta$ and $\varphi$. Then
\begin{itemize}
\item the stabilizer subgroup $\Stab{G}{x}$ is a subgroup of $\Stab{G}{\theta(x)} := \varphi^{-1}( \Stab{H}{\theta(x)})$, and
\item we can define a canonization map $\Can_G(x)$ in the following way:
\begin{enumerate}
\item Compute $\Can_H( \theta(x) ) = (\CF_H(\theta(x)), \TR_H(\theta(x)), \Stab{H}{\theta(x)})$  for some fixed canonization map $\Can_H$.
\item Compute $g \in \varphi^{-1}(\TR_H(\theta(x)))$ and $G' := \Stab{G}{g\theta(x)} = g\Stab{G}{\theta(x)}g^{-1}$.
\item Define $\Can_G(x) := (\CF_{G'}(gx), \TR_{G'}(gx) g, g^{-1}\Stab{G'}{gx}g )$  for some fixed canonization map $\Can_{G'}$.
\end{enumerate}
\end{itemize}
\end{thm}

\begin{exmp}
Let $G=(V, E)$ be a graph with finite vertex set $V:=\{1,\ldots, n\}$ and edges $E \subseteq \{\{x,y\} \mid x,y \in V: x \neq y\}$. There is a natural action of $\pi \in S_n$ on $G$ defined in the following way:

\[\pi G := (V, \pi E) := (V, \{ \{\pi(x),\pi(y)\} \mid \{x,y\} \in E\}).\]

If $\{x,y\} \in E$ we say that $y$ is a neighbor of $x$. Now, let $N(x)$ count the number of neighbors of $x$ and define the $S_n$-homomorphism $N(G) := (N(1), \ldots, N(n))$. The Homomorphism Principle tells us
\begin{enumerate}
\item to canonize this sequence under the action of the symmetric group $S_n$, for instance by sorting the sequence lexicographically.
\item If $\pi \in S_n$ is the corresponding permutation, we have to relabel the graph via the application of $\pi$ and
\item canonize the relabeled graph under the stabilizer $\Stab{S_n}{\pi N(G)}$.
\end{enumerate}
We may interpret the result of $N(\pi G)$ as some coloring on the vertices. In the following this coloring has to be preserved by the group action. This allows us to apply the Homomorphism Principle recursively since in the following we can count neighbors of a single color class as well.
\end{exmp}

\subsection{The lifting approach} \label{subsec:lifting_approach}

Let $H$ be a subgroup of $G$, short: $H \leq G$. A subset $T \subset G$ is called a \emph{right (left) transversal} of $H$ in $G$ if it is a minimal but complete set of right (left) coset representatives, i.e. $Ht \neq Ht'$ for all $t,t'\in T$ and $G = \bigcup_{t\in T} H t$.

\begin{prop} Let $G$ be a group acting on a totally ordered set $X$. Suppose that there is already some canonization $\Can_H$ available for $H < G$ and let $T$ be a right transversal of $H$ in $G$. Then, we can define the canonization map $\Can_G$ for $x \in X$ in the following way:
\begin{itemize}
\item $\CF_G(x) := \min_{t \in T} \CF_H(tx)$.
\item Let $t_1 \in T$ be a transversal element with $\CF_G(x) = \CF_H(t_1 x)$. Define $\TR_G(x) := \TR_H(t_1x) t_1$.
\item Let $t_1, \ldots, t_m \in T$ be those elements of $T$ which define a canonical form $\CF_H(t_i x) = \CF_G(x)$. The stabilizer $\Stab{G}{x}$ is generated by $\{t_it_1^{-1} \mid i=2,\ldots, m\}$ and $\Stab{H}{x}$.
\end{itemize}
\end{prop}

\begin{exmp}\label{ex:subgroup_canonization}
Like in the example above, let $G=(V, E)$ be a graph with $n$ vertices and let $H := \Stab{S_n}{1}$ be the stabilizer of $1 \in V$. Then, we may define the canonization of $G$ under the action of $S_n$ by comparing the canonical forms under the action of $H$ for the $n$  graphs derived by interchanging the vertices $1$ and $i$, $i=1,\ldots, n$.

For example, we may apply this approach if the number of neighbors is constant on $G$. Then, the separation of $1\in V$ allows us to color the vertex $1$ differently from all others and to count neighbors by colors again. This may result in different values and would allow us to define the canonization under $H$ with the help of the Homomorphism Principle.
\end{exmp}

\subsection{Partitions and Refinements}\label{subsec:part_refinement}
As we have seen in Example \ref{ex:subgroup_canonization} it makes sense to combine the methods of Subsections \ref{subsec:hom_principle} and \ref{subsec:lifting_approach}. The basic idea is to alternate between both methods and is known as the \emph{partition refinement method}: $\Can_G(x)$ is recursively computed via
\begin{enumerate}
\item the application of the Homomorphism Principle for a well-defined sequence of homomorphisms of group actions which may lead to a smaller stabilizer $G'$ and the element $x'=gx$.
\item If the group $G'$ is not trivial, we apply the lifting approach for a well-defined subgroup $H \leq G'$ and recursively continue
the computation of $\Can_H(tx')$ for $t\in T$ in a similar way. Otherwise, we just return $(x', \id_{G'}, \{\id_{G'}\})$.
\end{enumerate}

In this formulation, the different canonization processes $\Can_H(tx')$ for the right transversal elements $t \in T$ in the lifting approach are carried out independently. Of course, making use of some global information in this processes could further reduce the computational complexity. The partition refinement method also considers this problem as we will see later. For this reason, we will replace the above formulation by a backtracking approach.

Partition refinement methods are widely used in the canonization of combinatorial objects, for equivalence tests and automorphism group computations, for instance \cite{Feu09, FeuZ4, Leon:Code, Leon:PartitionRefinement, McKay:graphIsomorphism}. In most cases the authors restrict themselves to the action of the symmetric group or some special subgroups. The formulation above shows that the ideas presented there are also applicable for arbitrary groups.

Nevertheless, we will similarly formulate our algorithm only for the action of the symmetric group. The main reason for this restriction
is an easier description of the algorithm and some observations which we can only give in this special case. We will later see, that there is an action of the symmetric group in our problem, too.

A \emph{partition} of $\Set{n} := \{1,\ldots,n\}$ is a set $\mathfrak{p} = \{P_1,\ldots, P_l\}$ of disjoint nonempty subsets of $\Set{n}$ whose union is equal to $\Set{n}$. We call the subsets $P \in \mathfrak{p}$ \emph{cells} of the partition. Cells of cardinality $1$ are \emph{singletons} and the partition $\mathfrak{p}$ is \emph{discrete} if all its cells are singletons.
If all cells of a partition $\mathfrak{p}$ are intervals we call $\mathfrak{p}$ a standard partition. In the following we will always use upper-case letters for standard partitions.
The stabilizer

\begin{equation}\nonumber
S_{\mathfrak{p}} := \Stab{S_n}{\mathfrak{p}}:= \bigcap_{P \in \mathfrak{p}} \Stab{S_n}{P}
\end{equation}
of the (standard) partition $\mathfrak{p}$ is a \emph{(standard) Young subgroup} of $S_n$.
With $\Fixed{\mathfrak{p}} := \{i \in \Set{n} \mid \{i\} \in \mathfrak{p}\}$ we refer to those indices which define singletons of $\mathfrak{p}$, i.e. fixed points under the group action of $S_{\mathfrak{p}}$.

The partition $\mathfrak{p}$ is \emph{finer} than the partition $\mathfrak{p}'$ if each cell $P \in \mathfrak{p}$ is a subset of some cell of $\mathfrak{p}'$. We also call $\mathfrak{p}$ a \emph{refinement} of $\mathfrak{p}'$ and say that $\mathfrak{p}'$ is \emph{coarser} than $\mathfrak{p}$.

Differently to \cite{Leon:Code, Leon:PartitionRefinement, McKay:graphIsomorphism} our approach only uses standard partitions where the ordering of the cells is naturally defined by the elements they contain. This difference is due to the fact that we maintain a coset $S_{\mathfrak{P}}\pi$ of a standard Young subgroup, which is the key data structure in all algorithms, by the pair $(\mathfrak{P},\pi)$ instead.

\subsubsection{Backtrack tree}

Suppose there is the group action of a standard Young subgroup $S_{\mathfrak{P}_0}$ on a set $X$. For an element $x \in X$ we can compute its unique canonical form and its stabilizer using a backtrack procedure on the following search tree, see also Figure \ref{fig:backtracking}:
\begin{itemize}
\item The root node of the search tree is $(\mathfrak{P}_0,\id_{S_n})$ and we will apply a refinement on it as described below.
\item The nodes $(\mathfrak{P}, \pi)$ where $\mathfrak{P}$ is discrete define leaves of the tree.
\item Otherwise, i.e. in the case that $\mathfrak{P}$ is not discrete, we perform an \emph{individualization-refinement} step:
\begin{itemize}
 \item Choose a well-defined\footnote{The target cell selection is an $S_\mathfrak{P}$-invariant.} cell $P \in \mathfrak{P}$ which is not a singleton, called the \emph{target cell} and use the lifting approach for $S_{\mathfrak{P}'} := \Stab{S_\mathfrak{P}}{m} \leq S_\mathfrak{P}$ where $m = \min(P)$:
The refinement $\mathfrak{P}'$ of $\mathfrak{P}$ is derived by separating the minimal element $m\in P$, i.e. replace $P$ by $\{m\}$ and $P  \setminus \{m\}$. If $T$ is a right transversal of $S_{\mathfrak{P}'}$ in $S_{\mathfrak{P}}$, the $\vert T \vert = \vert P \vert$ different children of the actual node are constructed by applying the permutations $t \in T$.
\item A refinement of $\mathfrak{P}'$ for the node $(\mathfrak{P}', t\pi)$ could be computed via the application of the Homomorphism Principle using a fixed $S_{\mathfrak{P}'}$-homomorphism $f_{\mathfrak{P}'}: X \rightarrow Y$. We choose the action and the canonical forms in $Y$ in such a way that their stabilizers are again standard Young subgroups.

Let $\sigma := \TR_{S_{\mathfrak{P}'}}(f_{\mathfrak{P}'}(t \pi x))$ and $S_{\mathfrak{R}} := \Stab{S_{\mathfrak{P}'}}{f_{\mathfrak{P}'}(\sigma t \pi x)}$ be the result of the canonization of $f_{\mathfrak{P}'}(t \pi x)$.
The principle tells us that we have to canonize the element $\sigma t \pi x$ under the action of the group $S_{\mathfrak{R}}$ which is based again on further individualization-refinement steps.
\end{itemize}
\end{itemize}

\begin{figure}[tb]
\begin{center}
\includegraphics[width=0.7\textwidth]{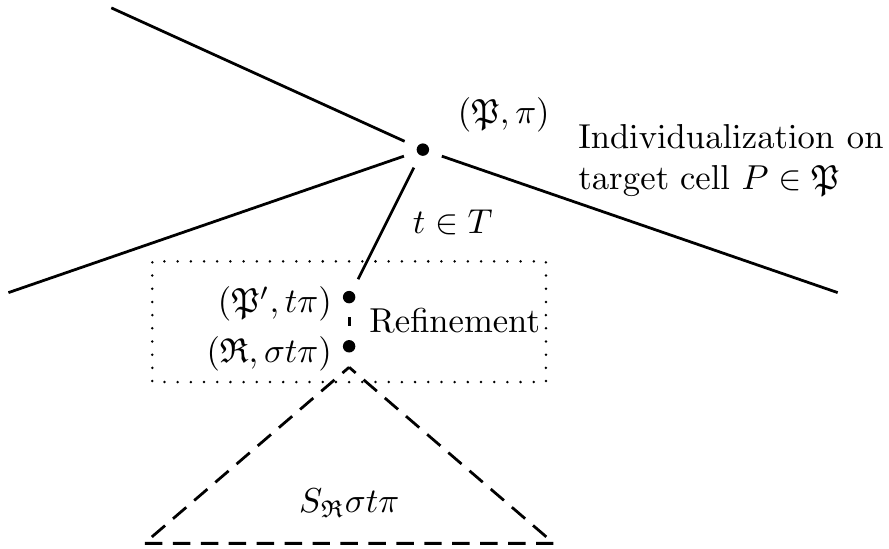}
\end{center}
\caption{Partition refinement backtrack tree} \label{fig:backtracking}
\end{figure}

Traversing this tree in a depth-first search manner corresponds to
the aforementioned alternating application of the Homomorphism
Principle and the lifting approach. We are able to define the
canonical representative $\Can_{S_\mathfrak{P}}(\pi x)$ if we have
visited all children of the node $(\mathfrak{P}, \pi)$. So far, the
canonization under the action of $S_{\mathfrak{P}'}$ are still
independent processes for different $t \in T$. The following
definition of a total ordering $\leq_\mathfrak{P}$ on $X$ allows us
to change this:

\[ x \leq_\mathfrak{P} y :\Longleftrightarrow f_{\mathfrak{P}'}(x) <_Y  f_{\mathfrak{P}'}(y) \vee \left( f_{\mathfrak{P}'}(x) =  f_{\mathfrak{P}'}(y) \wedge x \leq_X y \right)
\] where $\leq_X$ and $\leq_Y$ are still some arbitrary total orderings on $X$ and $Y$ respectively.

This ordering will be used in the lifting approach for the definition of the minimum and it allows us to prune the search tree, i.e. skip the canonization $\Can_{S_{\mathfrak{R}}}(t_2\pi x)$ in the following situation:
If the canonical form $\sigma_1 f_{\mathfrak{P}'}(t_1 \pi x)$ is smaller than the canonical form $\sigma_2 f_{\mathfrak{P}'}(t_2\pi x)$ in the Homomorphism Principle for two nodes arising in an individualization step, we prune the subtree rooted in $(\mathfrak{P}', t_2\pi)$.

\begin{rem}
For the sake of simplicity, we did not use homomorphisms of group
actions in the formulation of the refinement step and we restricted
the formulations to $S_{\mathfrak{P'}}$-homomorphisms
$f_{\mathfrak{P}'}$. The function $f_{\mathfrak{P}'}$ itself might
be a concatenation of several functions which allow a successive
application of the Homomorphism Principle. In this case, we adapt
the ordering $\leq_\mathfrak{P}$ such that we may prune the tree in
some intermediate step as well.
\end{rem}

In the case that $f_{\mathfrak{P}'}(x) =  f_{\mathfrak{P}'}(y)$, we may recursively use $x\leq_X y :\Longleftrightarrow x \leq_\mathfrak{R} y$, where $S_\mathfrak{R} := \Stab{S_{\mathfrak{P}'}}{x}$, to compare $x$ and $y$. Only in the case that $\mathfrak{R}$ is discrete, i.e. the corresponding node is a leaf, we use some fixed ordering on $X$. This also shows that we are not only allowed to compare the children of a fixed node among each other. In fact, we can prune a node $(\mathfrak{P}', \pi)$
of the search tree if there is a another node $(\mathfrak{P}', \sigma)$ on the same level having the same values for all $S_{\mathfrak{Q}}$-homomorphism $f_{\mathfrak{Q}}$ applied from the root down to these nodes and whose actual image $f_{\mathfrak{P}'}(\pi x)$ is larger than $f_{\mathfrak{P}'}(\sigma x)$.

\begin{thm}
Let $\mathfrak{D}$ denote the discrete partition of $\Set{n}$.
Suppose that $(\mathfrak{D}, \pi)$ is the last visited leaf of this
pruned search tree. The mapping

\begin{equation*}
\Can_G(x) := \left(\pi x, \pi, \left\{\pi^{-1} \pi_i \mid i \in
\Set{a}\right\}\right)
\end{equation*} defines a canonization, where
$\pi_1, \ldots, \pi_a$ are those permutations leading to all other
leaf nodes $(\mathfrak{D}, \pi_i)$ with $\pi_i x$ equal to $\pi x$.

\end{thm}

\begin{proof} Let $F := (f_{\mathfrak{P}_0}, \ldots, f_{\mathfrak{P}_r})$
be the sequence of $S_{\mathfrak{P}_i}$-homomorphism, $i=0,\ldots,
r$ applied in the generation process to the leaf node
$(\mathfrak{D}, \pi)$. The sequence $F$ defines a total ordering on
$X$:

\[ x \leq_{F} y :\Longleftrightarrow
(f_{\mathfrak{P}_0}(x), \ldots, f_{\mathfrak{P}_r}(x), x) \leq
(f_{\mathfrak{P}_0}(y), \ldots, f_{\mathfrak{P}_r}(y), y)
\]

The pruning ensures that all other leaves $(\mathfrak{D}, \sigma)$
of this search tree will lead to orbit elements $\sigma x$ with $\pi
x \leq_F \sigma x$. This shows that all leaves which were not
visited correspond to orbit elements that compare strictly larger
than $\pi x$. Therefore, $\pi x$ is the minimal orbit representative
of $S_{\mathfrak{P}_0} x$ under $\leq_{F}$. It is not difficult to
prove that starting this backtracking algorithm for some other
element $x' \in S_{\mathfrak{P}_0} x$ will lead to the same orbit
representative $\pi' x' = \pi x$.
\end{proof}

Obviously, the depth-first-search strategy only allows the pruning
of a subtree based on some partial information. In particular, we
have to explore  subtrees which will later be discarded. A
breadth-first-search strategy would avoid this behavior.
Nevertheless, there are more advantages of a depth-first search
approach: First of all, there are no storage limitations since we
only have to maintain the path from the root node to the actual
node. Furthermore, it is possible to discover automorphisms of the
object $x$,  since those can only be computed by the comparison of
leaf nodes. The group of known automorphisms of $x$ allows us to
perform a further pruning of the search tree. For this goal, we use
the methods described in \cite[Section 5.2]{Feu09}: We store the
subgroup of already known automorphisms $A \leq S_{n}$ by a complete
labeled branching. \cite[Lemma 5.9]{Feu09} now gives a simple
criterion if the coset $S_{\mathfrak{P}}\pi$, i.e. the subtree below
the node $(\mathfrak{P}, \pi)$, has to be traversed or not.

Finally, we would like to mention that \cite{Piperno} discusses a mixture of both strategies in the computation of a canonical form of a graph. We think that this approach might be applicable in our case as well, but we are not yet sure about all the consequences because we have to incorporate a second group action on $x$ at the same time, as we will see in the following section.

\section{A reformulation of the problem}\label{sec:reformulation}

The \emph{projective space} $\PSqk$ is the set of all subspaces of $\F_{q}^{k}$. As usual, we call the one-dimensional subspaces \emph{points} and the $(k-1)$-dimensional subspaces \emph{hyperplanes}.
Let $\Aut(\F_q)$ denote the automorphism group of $\F_q$. Recall that any automorphism $\alpha \in \Aut(\F_q)$ is a power $\tau^a$ of the \emph{Frobenius automorphism}  $\tau: x \mapsto x^p$. It applies to vectors and matrices element-wise.
The set of all semilinear mappings, i.e. the \emph{general semilinear group} $\GammaL{k}(q) := \GL{k}(q)\rtimes \Aut(\F_q)$, decomposes as a semidirect product with multiplication $(A,\tau^a) (B, \tau^b) := (A \: \tau^a (B), \tau^{a+b})$.

\begin{rem} Let $N, H$ be arbitrary groups and $\varphi: H \mapsto \Aut(N)$ be a homomorphism between the group $H$ and the automorphism group of $N$. Although the multiplication of elements $(n_1, h_1), (n_2, h_2) \in N\rtimes_\varphi H$ depends on the choice of $\varphi$, i.e.
$ (n_1, h_1)(n_2, h_2) := (n_1\varphi(h_1)(n_2), h_1h_2)$, we will
not give the exact definition of the homomorphism $\varphi$ when
introducing semidirect products of groups in the following. We
believe that the right choice of $\varphi$ can always be observed
from the context.
\end{rem}

There is a natural action of $\GammaL{k}(q)$ on the projective space $\PSqk$ from the left, i.e.

\begin{align*}
\GammaL{k}(q) \times \PSqk \rightarrow \PSqk \\
\left( (A,\tau^a), \U\right) \mapsto A \tau^a(\U).
\end{align*}
Since this action is not faithful, one may also factor out the kernel resulting in the action of $\PGammaL{k}(q) := \GammaL{k}(q) / \mathbb{F}_q^\ast$ on $\PSqk$, where $\F_q^{\ast}$ denotes the multiplicative group of $\F_q$. We use both groups and both actions interchangeably.

The goal of this paper is the description of a practical canonization algorithm for a given sequence of (multi-) sets $\C = (\C_1, \ldots, \C_m)$, with $C_i \subseteq \PSqk$ under the action of $\GammaL{k}(q)$. The stabilizer subgroup

\[
\Aut(\C) := \Stab{\GammaL{k}(q)}{\C} := \bigcap_{i=1}^{m}  \{  (A,\tau^a) \in \GammaL{k}(q) \mid \forall\ \U \in \C_i : (A,\tau^a)\U \in \C_i\}
\]
is computed by the algorithm at the same time without any additional effort. We apply this algorithm to solve canonization problems in coding theory, see Section \ref{sec:CodingTheory}.

The remaining part of this section deals with further modifications of the given sequence $\C$ we could make:
\begin{itemize}
\item We may assume that the multisets $\C_i$ are in fact disjoint subsets. Otherwise, we could distinguish the occurring subspaces by their sequence of multiplicities, which leads to a sequence of disjoint subsets $\C'$. This defines an $\GammaL{k}$-homomorphism and the stabilizer of $\C'$ acts trivially on $\C$.
\item The action of $\GammaL{k}(q)$ preserves the dimension of any $\U \in \PSqk$. Hence, asking for a canonization algorithm for a set $\mathcal{C}_i$ is equivalent to ask for a canonization of the sequence $(\{\U \in \mathcal{C}_i \mid \dim(\U) = s \})_{s=0,\ldots, k}$. Therefore, we can assume that all elements of a subset $\C_i$ have fixed dimension $0\leq s_i \leq k$.
\item If some subset $\C_i$ is empty or equal to $\{ \U \in \PSqk \mid \dim(U) = s\}$, i.e. the subset of all $s$-dimensional subspaces of $\F_q^k$, for some $s=0,\ldots,k$, we can remove $\C_i$ from the sequence since the action of $\GammaL{k}(q)$ on this subset is trivial. Therefore we could suppose that $1\leq s_i \leq k-1$.
\item The union $\bigcup_{i=1}^{m} \C_i$ spans the whole space, otherwise we would be able to solve the problem in a smaller ambient space $\F_q^{k'}, k'<k$.
\end{itemize}
In the following, we suppose that the sequence $\C = (\C_1, \ldots, \C_m)$ and therefore also the parameters $q, k, m, n_i, s_i, n=\sum_{i=1}^m n_i$ will be fixed.

\begin{prop}\label{prop:actionDual}
For a given subspace $\U \in \PSqk$ let $\U^\perp := \{ v \in \F_q^k  \mid v^T u=0\}$ be its dual subspace.
The dual subspace of $(A, \tau^a)\U$ for $(A, \tau^a) \in \GammaL{k}(q)$ is equal to $({A^{-1}}^T, \alpha)\U^\perp$.
\end{prop}

\begin{proof}
Let $\U \in \PSqk$, $(A, \tau^a) \in \GammaL{k}(q)$ and $u \in U, v \in \U^\perp$ be arbitrary. The equation

\begin{align*}
(({A^{-1}}^T, \tau^a)v)^T (A,\tau^a)u = \tau^a(v)^T A^{-1} A\tau^a(u)
= \tau^a(v^T u) = \tau^a(0) = 0
\end{align*}
shows that $({A^{-1}}^T, \tau^a)\U^\perp \subseteq \left((A, \tau^a)\U \right)^\perp$. But both subspaces have dimension $n-\dim(\U)$ and hence must be equal.
\end{proof}

\begin{rem}
If we define $\C_i^\perp := \{\U^\perp \mid \U \in \C_i\}$ and
$\C^\perp := (\C_1^\perp, \ldots, \C_m^\perp)$ then we may also
canonize $\C^\perp$, i.e. compute $\Can_{\GammaL{k}(q)}(\C^\perp)$,
and define the canonical form $\overline{\CF}_{\GammaL{k}(q)}(\C) :=
\left(\CF_{\GammaL{k}(q)}(\C^\perp)\right)^\perp$. The automorphism
group of $\C$ is equal to

\[\left\{ ({A^{-1}}^T, \tau^a) \mid (A, \tau^a)\in  \Aut(\C^\perp) \right\}.\]
This transformation will always be applied if we suppose that the computational effort of computing $\Can_{\GammaL{k}(q)}(\C^\perp)$ is less expensive than the computation of $\Can_{\GammaL{k}(q)}(\C)$.
\end{rem}

Let $\F_q^{k\times n, s}$ denote the set of $k\times n$-matrices of
rank $s$. The algorithm we are going to develop is a generalization
of the canonization algorithm for linear codes, see \cite{Feu09,
FeuZ4} and Section \ref{subsec:linCodes} for a short summary.
Instead of working on linear codes directly, the problem is
transferred to generator matrices of linear codes, i.e. matrices
whose rows form an $\F_q$-basis of the linear code. Two matrices
$\Gamma, \Gamma' \in \F_q^{k\times n,k}$  generate equivalent codes,
if their orbits under the group action of $ ( \GL{k}(q) \times
{\F_q^{\ast}}^n ) \rtimes (S_n \times \Aut(\F_q))$ are the same. It
is a well-known fact \cite[9.1.2]{Kerber} that there is a one-to-one
correspondence of these orbits and the orbits of $\GammaL{k}(q)$ on
multisets of at most $n$ points in the projective space, which span
a vector space of dimension $k$. Therefore, the canonization
algorithm for linear codes already solves the canonization problem
for any multiset of points. A closer look reveals that the algorithm
similarly transfers the multiset to a sequence of disjoint sets of
points.

Now, represent the element $\U \in \C_i$ by some matrix $U \in
\F_q^{k \times s_i, s_i}$ whose columns generate $\U$. The set of
all matrices generating $\U$ in this regard is equal to the orbit
$\GL{s_i}(q)U := \{ UA^T \mid A \in \GL{s_i}(q)\}$. Analogously, we
can identify the set $\C_i := \{\U_1,\ldots, \U_{n_i}\}$ with the
orbit of $S_{n_i}$ on $(\U_1,\ldots, \U_{n_i})$. In summary, there
is a natural one-to-one correspondence between the set $\C_i$ and
the orbit of $(U_1, \ldots, U_{n_i})$ under the action of
$\GL{s_i}(q)^{n_i} \rtimes S_{n_i}$. This semidirect product is
equal to the wreath product $\GL{s_i}(q) \wr S_{n_i}$.

In the case that $s_i=1$ we know that $\GL{1}(q) = \F_q^\ast$ and the group $\GL{1}(q)^{n_i} \rtimes S_{n_i} = {\F_q^\ast} \wr S_{n_i}$ is isomorphic to the group of $\F_q^\ast$-monomial matrices, i.e. the set of permutation matrices whose nonzero entries got replaced by elements from $\F_q^{\ast}$. In this regard, we can view the wreath product $\GL{s_i}(q) \wr S_{n_i}$ as the group of $\GL{s_i}(q)$-monomial matrices by replacing the nonzero entries of a permutation matrix by arbitrary elements from $\GL{s_i}(q)$ and the zero entries by $(s_i\times s_i)$-zero matrices.

Altogether, we can identify the sequence $\C$ with the orbit of

\begin{equation*}
{\left(U_1^{(i)}, \ldots, U_{n_i}^{(i)}\right)}_{i \in \Set{m}} \in \prod_{i=1}^{m}\left(\F_q^{k \times s_i, s_i}\right)^{n_i}
\end{equation*}
under the action of $\prod_{i=1}^{m} (\GL{s_i}(q)^{n_i} \rtimes S_{n_i})$ which could be interpreted as the group of block diagonal matrices whose $m$ nonzero blocks are equal to $\GL{s_i}(q)$-monomial matrices.

Finally, taking the action of $\GammaL{k}(q)$ into account we have
to canonize the sequence $\left( ( U_1^{(i)}, \ldots, U_{n_i}^{(i)}
)\right)_{i\in \Set{n}}$ under the action of

\begin{equation*}
\left( \GL{k}(q) \times  \prod_{i=1}^{m} (\GL{s_i}(q)^{n_i} \rtimes S_{n_i}) \right) \rtimes \Aut(\F_q),
\end{equation*}
where the action is defined as follows:\\
Let $\left(A, {\left(B_{1}^{(i)},\dots,B_{n_i}^{(i)}, \pi^{(i)}\right)}_{i \in \Set{m}}, \tau^a \right)$ be an element of the acting group then

\begin{align*}
&\left(A, {\left(B_{1}^{(i)},\dots,B_{n_i}^{(i)}, \pi^{(i)}\right)}_{i\in \Set{m}}, \tau^a \right)  {\left(U_1^{(i)}, \ldots, U_{n_i}^{(i)}\right)}_{i\in \Set{m}} \\
:=& {\left(A\tau^a \left(U_{{\pi^{(i)}}^{-1}(1)}^{(i)}\right){B^{(i)}_{1}}^T, \ldots, A\tau^a\left(U_{{\pi^{(i)}}^{-1}(n_i)}^{(i)}\right) {B^{(i)}_{n_i}}^T\right)}_{i\in \Set{m}}.
\end{align*}

In order to apply the methods developed in the Subsection \ref{subsec:part_refinement}, we change the order in which we compose the group:

\begin{align*}
&\left(\left( \GL{k}(q) \times  \prod_{i=1}^{m} \GL{s_i}(q)^{n_i} \right) \rtimes \Aut(\F_q) \right) \rtimes \prod_{i=1}^{m} S_{n_i}\\
\simeq &\left(\left( \GL{k}(q) \times  \prod_{i=1}^{m} \GL{s_i}(q)^{n_i} \right) \rtimes \Aut(\F_q) \right) \rtimes S_{\mathfrak{P}_0}
\end{align*}
and replace the permutational part of this group using the standard Young subgroup $S_{\mathfrak{P}_0}$ to the partition $\mathfrak{P}_0 := \{\{1,\ldots, n_1\}, \ldots, \{n-n_m+1,\ldots, n\}\}$ of $\Set{n}$.

A final reformulation of our problem will be given in Section \ref{sec:Algorithm} since we would like to motivate the algorithm with the observations given in \cite{Feu09}. It will show that we could observe a homomorphic group action of the symmetric group $S_{\mathfrak{P}_0}$. This allows us to apply the ideas developed in Subsection \ref{subsec:part_refinement}.

As we have seen above, the comparison of the objects we are working
with plays a central role in the canonization process. In our case,
if nothing else is stated, we will suppose that $\F_q$ is totally
ordered such that $0<1 \leq \mu$ for all $\mu \in \F_q^\ast$. Then
we can totally order the set of $k\times n$-matrices by interpreting
them as lexicographically ordered sequences of colexicographically
ordered column vectors.

Furthermore, we will access submatrices of a matrix $U \in \Mat{k}{n}$ in the following way:
\begin{itemize}
\item $U_{\ast,i}$ denotes the $i$-th column of $U$. Similarly, we write $U_{i, \ast}$ for the $i$-th row.
\item For a sequence $I:=(i_1,\ldots, i_m)$ of indices $i_j \in \Set{n}$ we write $U_{\ast, I} := (U_{\ast, i_1}, \ldots, U_{\ast,i_m})$ for the projection of the matrix onto the columns given by $I$. We also use this notation for the set $I:=\{i_1, \ldots, i_m\}$ which should be interpreted as the lexicographically ordered sequence of its elements.
\item Finally, if $J$ is a sequence of indices in $\Set{k}$, then $U_{J, \ast}$ denotes a similar access to the rows of $U$ and
$U_{J, I} := (U_{J, \ast})_{\ast,I}$.
\end{itemize}

\section{Coding theory}\label{sec:CodingTheory}
Let $(M_1,d_1)$ and $(M_2, d_2)$ be two metric spaces. A map $\iota: M_1
\rightarrow M_2$ is an \emph{isometry} if it respects distances, i.e.
$d_2(\iota(x), \iota(y)) = d_1(x,y)$ for all $x,y \in M_1$.


\subsection{Random network codes}
The \emph{subspace distance} is a metric on the projective space $\PSqk$ given by

\begin{align*}
d_S(\mathcal{U},\mathcal{V}) := \dim(\U + \mathcal{V}) - \dim(\U \cap \mathcal{V})
= \dim(\U) + \dim(\mathcal{V}) - 2\dim(\mathcal{U}\cap
\mathcal{V})
\end{align*}
for any $\mathcal{U},\mathcal{V} \in \PSqk$. It is a suitable distance for coding over the operator channel
using so-called \emph{random network codes} $\C \subset \PSqk$, see \cite{ko08}.

Obviously, the action of an element of the general semilinear group preserves the subspace distance.
On the other hand, \cite{traut11} showed that $\PGammaL{k}(q)$ is isomorphic to the group of isometries on $\PSqk$ which preserve the dimension of each element in $\PSqk$. The dimension is another basic property of a codeword which should be preserved, too. Therefore, it makes sense to define equivalence of random network codes by means of this group action. It shows that the canonization of random network codes is a special case of our algorithm for sequences of length one.

\subsection{Additive codes}
An \emph{$\F_q$-linear block code} over the alphabet $\F_{q^{s}}, s \geq 1$ is an $\F_q$-linear subset of $\F_{q^{s}}^n$ equipped with the usual Hamming distance $d_{\operatorname{Ham}}$. Additive codes with $s=1$ are classical linear codes. For $s>1$, those codes are sometimes also called \emph{additive codes}. They gained more and more interest in the past years since for example self-orthogonal additive codes over $\F_{q^{2}}$ could be used for quantum error-correction, see \cite{Ashikhmin_Knill}.

With an $\F_q$-linear representation of the elements of $\F_{q^{s}}$ in $\F_{q}^{s}$, the $\F_q$-linear code can be represented by a generator matrix with entries in $\F_{q}$. Let $T: \F_{q^{s}} \rightarrow \F_{q}^{s}$ denote the corresponding $\F_q$-linear mapping.
Defining the distance

\[d_{\operatorname{Ham}_s}(x,y) := d_{\operatorname{Ham}}(T^{-1}(x), T^{-1}(y)) = \begin{cases} 0 & x=y \\1, & else \end{cases} \textnormal{ for } x,y \in \F_{q}^{s}\] and extending this definition as usual to $\F_{q}^{sn}$ we are able to find all isometries on $\F_{q^{s}}^n$ mapping
$\F_q$-linear codes onto $\F_q$-linear codes:

\begin{itemize}
\item The multiplication of $\F_{q}^{s}$ by an invertible matrix $A \in \GL{s}(q)$ defines an $\F_q$-linear isometry.
\item The same holds for the permutation of the $n$ components of $\F_{q}^{sn}$.
\item The element-wise application of an automorphism of $\F_q$ defines an isometry on $\F_{q}^{sn}$, which maps $\F_q$-linear codes onto $\F_q$-linear codes.
\end{itemize}
Altogether, this defines a group action of $\GL{s}(q)^n \rtimes (S_n \times \Aut(\F_q))$ on the set of $\F_q$-linear subsets of $\F_{q^{s}}^n$. Since isometries are injective, we could also restrict this action to act on subsets $C$ with $\dim_{\F_q}(C) = k$. Each such subset $C$ could be represented by a generator matrix $\Gamma \in \F_q^{k \times sn, k}$. The set of all generator matrices of $C$ is equal to the orbit $\GL{k}(q)\Gamma$. Hence, we are interested in the canonization of a generator matrix under the group action of

\begin{equation*}
\left(\GL{k}(q) \times \GL{s}(q)^n \right) \rtimes (S_n \times \Aut(\F_q)).
\end{equation*}
Since every $s$ consecutive columns may define the same subspace, the code could be identified with a multiset of subspaces of $\F_q^k$.

\section{The Algorithm}\label{sec:Algorithm}

In this section we develop a practical algorithm which computes the automorphism group and a canonical form of a given sequence
$\C = (\C_1,\ldots, \C_m)$. In Section \ref{sec:Canonization} we have seen why it is useful to combine the searches for both tasks.

\subsection{The algorithm for linear codes revisited}\label{subsec:linCodes}
First of all, we want to motivate our procedure by a reformulation
of the canonization algorithm in \cite{Feu09} for linear codes using
the language developed in Section \ref{sec:Canonization}. Some of
the ideas we are going to introduce are based on \cite{FeuZ4} which
gives a more detailed description of the backtracking approach. The
algorithm that we are going to develop in the subsequent subsections
can be seen as a natural generalization of the one for linear codes.

We first observe that we could compute a canonical form of a $k$-dimensional linear code $C$
with generator matrix $\Gamma \in \Mat{k}{n,k}$ using the ideas presented in Section \ref{subsec:part_refinement}:
For simplicity, let $\Gsl:= \left( \GL{k}(q) \times {\F_q^\ast}^n \right) \rtimes \Aut(\F_q)$. The group $S_{\mathfrak{P}_0}$
acts on the set of orbits $\Gsl \backslash\!\!\backslash \Mat{k}{n,k} := \{ \Gsl \Gamma' \mid \Gamma' \in \Mat{k}{n,k} \}$ and we can define a homomorphism of group actions

\begin{equation*}
\left(
\begin{array}{rrclcrrcl}
\theta:& \Mat{k}{n,k} &\rightarrow& \Gsl \backslash\!\!\backslash \Mat{k}{n,k} & , &
\varphi:& \Gsl \rtimes S_{\mathfrak{P}_0} &\rightarrow& S_{\mathfrak{P}_0} \\
& \Gamma &\mapsto & \Gsl\Gamma &&& (g, \pi) &\mapsto& \pi
\end{array}
\right).
\end{equation*}

Before we provide the details of the canonization $\Can_{S_{\mathfrak{P}_0}}(\Gsl \Gamma )$ we explain how to define
$\Can_{\Gsl \rtimes S_{\mathfrak{P}_0}}$ using the Homomorphism Principle:
First of all, in \cite{Feu09} it is observed that there is a direct and efficient canonization algorithm $\Can_{\Gsl}$ for the action of $\Gsl$ on $\Mat{k}{n}$. Let $\pi = \TR_{S_{\mathfrak{P}_0}}(\Gsl \Gamma )$ and $\Gsl \rtimes H = \Stab{\Gsl \rtimes S_{\mathfrak{P}_0}}{\pi \Gsl \Gamma}$ be the result of the canonization. Since we know that $\left(\Gsl \rtimes H \right) \pi \Gamma = \Gsl\pi\Gamma$ it remains to define
$\CF_{\Gsl \rtimes H}(\pi\Gamma) := \CF_{\Gsl}(\pi\Gamma)$ and the transporter
$\TR_{\Gsl \rtimes H}(\pi\Gamma) := (\TR_{\Gsl}(\pi\Gamma), \id)$. There is also a simple way to compute the automorphism group
\linebreak $\Stab{ \Gsl \rtimes H }{\Gamma}$ using the canonization under the action of $\Gsl$. The details are left to the reader.
We will later see that all necessary data is already computed in the computation of $\Can_{S_{\mathfrak{P}_0}}(\Gsl \Gamma )$.

\subsubsection{Backtrack search}
This shows that we are able to give a practical canonization algorithm for linear codes if we are able to give a practical algorithm for the computation of $\Can_{S_{\mathfrak{P}_0}}(\Gsl \Gamma )$. This algorithm will be based on the partition refinement idea. In this algorithm, it is necessary to compare the leaves of the backtrack search tree. Therefore, we have to define a total ordering on $\Gsl \backslash\!\!\backslash \Mat{k}{n,k}$: Let $\Can_{\Gsl}$ be a canonization algorithm for the action of $\Gsl$ on the set $\Mat{k}{n,k}$, then we may order the orbits via the ordering on their canonical forms.

It remains to give the $S_{\mathfrak{P}}$-homomorphisms which will
be applied to a node $(\mathfrak{P},\pi)$. The first is closely
related to $\Can_{\Gsl}$. Let $\FixedSeq{\mathfrak{P},\pi}$ be the
sequence of elements of $\Fixed{\mathfrak{P}}$ in the order they
appeared as singletons in the refinement process $\mathfrak{P}_0
\geq \ldots \geq \mathfrak{P}$ leading to this node
$(\mathfrak{P},\pi)$. We say that a matrix
$\Gamma^{(\mathfrak{P},\pi)}$ is a \emph{semicanonical
representative} of the node $(\mathfrak{P},\pi)$ if
$\Gamma^{(\mathfrak{P},\pi)} \in \Gsl\pi\Gamma$ and

\begin{equation*}
\left(\Gamma^{(\mathfrak{P},\pi)}\right)_{\ast,
\FixedSeq{\mathfrak{P},\pi)}} \leq \left(\Gamma'\right)_{\ast,
\FixedSeq{(\mathfrak{P},\pi}} \textnormal{ for all } \Gamma' \in
\pi\Gsl\Gamma = \Gsl\pi\Gamma.
\end{equation*}

\begin{prop}
The projection $\pi\Gsl\Gamma \mapsto
\left(\Gamma^{(\mathfrak{P},\pi)}\right)_{\ast,
\FixedSeq{\mathfrak{P},\pi}}$ is $S_{\mathfrak{P}}$-invariant.
\end{prop}

This invariant is applied immediately after each individualization
step and after each refinement which leads to a new singleton in the
partition $\mathfrak{P}$. Since it is an invariant, it will not
refine the partition $\mathfrak{P}$. But, it will give us the
possibility to prune the search tree.

For a child $(\mathfrak{R},\sigma\pi)$ of $(\mathfrak{P},\pi)$ the
semicanonical representative could be easily computed from the
semicanonical representative of $(\mathfrak{P},\pi)$. For this
computation, we only need to know the stabilizer
$\Inn^{(\mathfrak{P}, \pi)}\leq \Gsl$ of
$\left(\Gamma^{(\mathfrak{P},\pi)} \right)_{\ast,
\FixedSeq{\mathfrak{P},\pi}}$, where the action is defined by

\begin{equation*}
(A, b, \tau^a) \Gamma' := (A, {(b_j)}_{j \in
\FixedSeq{\mathfrak{P},\pi}}, \tau^a) \Gamma' \textnormal{ for all }
\Gamma' \in \Mat{k}{\vert \Fixed{\mathfrak{P}} \vert}.
\end{equation*}
For more details see \cite{Feu09}.

Therefore, it makes sense to add this
data to the nodes of the backtrack tree. We modify the nodes $(\mathfrak{P}, \pi)$ of Section \ref{subsec:part_refinement}, such that the algorithm additionally maintains the orbit $\Gsl\pi\Gamma$ by the pair $(\Gamma^{(\mathfrak{P},\pi)}, \Inn^{(\mathfrak{P}, \pi)})$.
We will call the action by $\Inn^{(\mathfrak{P}, \pi)}$ the \emph{inner group action} and $\Inn^{(\mathfrak{P}, \pi)}$ the \emph{inner stabilizer}. The computation of the semicanonical representative of $(\mathfrak{P}, \pi)$ will also be called the \emph{inner minimization process}.

\begin{rem}
The computation of a semicanonical representative itself could be
seen as an application of the Homomorphism Principle applied to the
$\Gsl$-homomorphism $\Gamma \mapsto \Gamma_{\ast,
\FixedSeq{\mathfrak{P},\pi}}$. If $\mathfrak{P}$ is discrete the
semicanonical representative defines a canonical form of the orbit
$\Gsl\pi\Gamma$.
\end{rem}

\begin{rem}\label{rem:backtrack_interpretation}
There is also a second interpretation of the subtree below some node $(\mathfrak{P}, \pi, \Gamma^{(\mathfrak{P}, \pi)}, \Inn^{(\mathfrak{P}, \pi)})$. It could be identified as the canonization of $\Gamma^{(\mathfrak{P}, \pi)}$ under the action of
$\Inn^{(\mathfrak{P}, \pi)} \rtimes S_\mathfrak{P}$. For the root node, the group $\Inn^{(\mathfrak{P}_0, \id)} \rtimes S_{\mathfrak{P}_0}$ is equal to $\Gsl \rtimes S_{\mathfrak{P}_0}$, i.e. the action we are actually interested in.
We have motivated this backtracking with the canonization under $S_{\mathfrak{P}_0}$, since
\begin{itemize}
\item we already proved its correctness in Section \ref{subsec:part_refinement},
\item the test on the group of known automorphisms is restricted to standard Young subgroups, and
\item the homomorphism of group actions we are going to apply should have this special structure, i.e. they will be either $S_{\mathfrak{P}}$-homomorphisms or equal to the inner minimization process. This will avoid the occurrence of complex subgroups of  $\Gsl \rtimes S_{\mathfrak{P}_0}$ for which it would be difficult to define appropriate homomorphisms of group actions in the refinement steps. Furthermore, the computation of the transversal $T$ in an individualization step would become more complicated, too.
\end{itemize}
\end{rem}

Apart from some further $S_{\mathfrak{P}}$-homomorphisms which make
use of $(\Gamma^{(\mathfrak{P}, \pi)}, \Inn^{(\mathfrak{P}, \pi)})$,
there is another very important $S_{\mathfrak{P}}$-homomorphism used
to derive further refinements. In particular, this
$S_{\mathfrak{P}}$-homomorphism works already very well on nodes on
the first levels of the backtrack tree. We are going to generalize
it in the following. Again, we will give a more general description
of this function than given in \cite{Feu09, FeuZ4}. The basic idea
is a modification of Leon's algorithm, \cite{Leon:Code}, for the
computation of the automorphism group of a linear code:

\subsubsection{Leon's invariant set of codewords}
Suppose that $W := \{c^{(1)}, \ldots, c^{((q-1)h)}\} \subseteq C$ is
a set of codewords, which is invariant under the automorphism group
of $C$. In fact, if one wants to apply Leon's algorithm,
\cite{Leon:Code}, to test two linear codes for equivalence the
mapping $C \mapsto W$ has to be an $(\F_q^{\ast})^n \rtimes (S_n
\times \Aut(\F_q))$-homomorphism. For simplicity suppose that $W$ is
formed by all words of minimal nonzero weight $d$. This set is once
computed for the root node and fixed for the whole backtracking
algorithm.

For each codeword $c^{(j)}$ there is a well-defined information word $v^{(j)} \in \F_{q}^{n}$ such that $v^{(j)}\Gamma = c^{(j)}$. Since $W$ is closed under scalar multiplication, we may restrict ourselves to the set $\overline{W}$ of projective representatives and define by
$\H := \{ v^{(j)} \mid j=1, \ldots, h\} := \{ v \in \F_q^n \mid v\Gamma \in \overline{W}\}$ the corresponding set of information words.

The standard inner product of $\langle v^{(j)},\Gamma_{\ast, i} \rangle = {v^{(j)}}^T\Gamma_{\ast, i}$ is equal to $c_i^{(j)}$, the $i$-th coordinate of the vector $c^{(j)}$. Therefore, the set $\H$ is the well defined subset of all normal vectors of those hyperplanes containing exactly $n-d$ points $\Point{\Gamma_{\ast, i}}$.

This gives a bipartite, vertex-colored subgraph of the subspace lattice of the projective space, whose vertices are labeled by $\Set{n}$ and $\{n+1,\ldots, n+h\}$ respectively. Initially, the colors only distinguish vertices by dimension and in the case of points additionally by the cell they are contained in. Since the action by $\GammaL{k}(q)$ obviously preserves the graph structure, this graph is independent from the actual representative of $\Gsl\Gamma$. Furthermore, it is well-defined up to the action of $S_h$ on the vertices $\{n+1,\ldots, n+h\}$. The permutation of the columns of the generator matrix results in a relabeling of the vertex set $\Set{n}$.

Cell-wise counting of neighbors for each vertex allows us to define an $S_{\mathfrak{P}_0}\times S_h$-homomorphism and hence to apply the Homomorphism Principle to refine the partition (coloring) of $\Set{n+h}$. Since the projection on the first $n$ components also obviously defines an $S_{\mathfrak{P}_0}$-homomorphism, this could be also seen as a refinement of the root node of the backtrack search tree. The resulting permutation in the application of the Homomorphism Principle gives a simultaneous relabeling of the graph and a permutation of the columns of the generator matrix.

This homomorphism on the incidence graph is also used in the
computation of a canonical form of an arbitrary graph
\cite{McKay:graphIsomorphism}. We furthermore observe that the finer
partitions $\mathfrak{R}_0$ of $\Set{n}$ and $\mathfrak{Q}_0$ of
$\Set{h}$ allow us to call this invariant iteratively. Hence, the
nodes of the backtrack tree additionally maintain a partition
$\mathfrak{Q}$ of the set $\Set{h}$. Furthermore, instead of storing
the relabeled graph, we just maintain a second permutation $\sigma
\in S_h$ which stores the relabeling on the vertices $\{n+1,\ldots,
n+h\}$. Any refinement on the partition $\mathfrak{P}$ during the
backtrack search gives us the possibility to restart this refinement
process on the relabeled, newly colored incidence graph again.

\subsection{Preprocessing and the backtrack tree}\label{subsec:backtracking}

In the same manner, we start the algorithm for the sequence $\C$ by some preprocessing routine, which is in fact an $\GammaL{k}(q)$-homomorphism computing a set $\H$ of hyperplanes:
For each hyperplane $\Point{v}^\perp \in \PSqk$ we may compute the number of elements it contains from each cell of $\mathfrak{P}_0$. This results in a unique partition of the set of all hyperplanes. We choose a well-defined subset $\H := \bigcup_{i=1}^{m'} \H_i$, which is a union of the blocks of this partition in such a way that it contains $k$ linearly independent normal vectors. This set should also be reasonably small.

\begin{exmp} \label{ex:preprocessing}
Let $\C := \{ \U_1, \U_2, \U_3 \} \subset \PS{3}{4}$ with $\U_i$ generated by $U_i$:

\begin{equation*}
\begin{array}{ccc}
U_1 := \begin{pmatrix}
1 & 0 & 0 & 0\\
0 & 1 & 0 & 0
        \end{pmatrix}^T &
U_2 := \begin{pmatrix}
0 & 0 & 1 & 0\\
0 & 0 & 0 & 1
        \end{pmatrix}^T&
U_3 := \begin{pmatrix}
1 & 0 & 0 & 0\\
0 & 0 & 1 & 0
        \end{pmatrix}^T
\end{array}
\end{equation*}
The hyperplanes $\H_1 := \left\{\Point{(0,1,0,0)^T}^\perp, \Point{(0,0,0,1)^T}^\perp \right\}$ contain exactly $2$ elements of $\C$, whereas the elements $\Point{(0,1,0,\nu)^T}^\perp$, $\Point{(\mu,1,0,0)^T}^\perp$ and $\Point{(0,0,\mu,1)^T}^\perp$ for $\nu \in \F_3^\ast, \mu \in \F_3$  contain a single element of $\C$. They form the set $\H_2$. The remaining hyperplanes contain none of the elements from $\C$. We choose the set $\H = \H_1 \cup \H_2$ since it is the smallest set spanning the whole vector space.
\end{exmp}

In difference to the algorithm for linear codes described above, we do not only use this set as a tool for the refinement of a node. In fact, we append the sequence $(\H_1, \ldots, \H_{m'})$ to $\C$ and it plays an active role in the construction of the backtrack search tree, as we will see later. We are allowed to do so because of the following proposition:

\begin{prop}
Let $G$ be a group which acts on $X$ and $f:X\rightarrow Y$ a
$G$-homomorphism. Suppose there is a canonization algorithm
$\Can_G((x,f(x)))$ for the action of $G$ on $\{(x, y) \mid x \in X,
y \in Y\}$. Then, this defines a canonization algorithm on $X$ via

\begin{equation*}
\Can_G(x) := \left( \CF_G((x, f(x)))_1, \TR_G((x, f(x))), \Stab{G}{(x, f(x))} \right).
\end{equation*}
\end{prop}

One reason for this decision is the fact that the partition $\mathfrak{Q}$ of $\Set{h}, h:= \vert \H \vert$  allows us to perform the individualization step on a cell with smaller cardinality and hence results in a smaller branching factor, which we see as an advantage. On the other hand, we realized that it is much more difficult to give an efficient canonization algorithm for the action of

\begin{equation*}
\left( \GL{k}(q) \times  \prod_{i=1}^{m} \GL{s_i}(q)^{n_i} \right) \rtimes \Aut(\F_q)
\end{equation*}
on the sequence of subspaces whose coordinates are fixed by $S_{\mathfrak{P}}$.

Let $\mathfrak{P}$ be a refinement of $\{\{1,\ldots, n\}, \{n+1,\ldots, n+h\}\}$. In the following we use $\mathfrak{P}_\C$ to refer to the corresponding partition of $\Set{n}$ and $\mathfrak{P}_\H$ to refer to the partition of the set $\Set{h}$ arising from the partition of the last $h$ coordinates given by $\mathfrak{P}$.

Instead of representing the hyperplane $H \in \H$ using a matrix $U \in \F_q^{k\times (k-1), k-1}$ we use its dual space which could be represented by a single vector $v \in \F_q^k$ -- a normal vector of $H$. Since $(A,\alpha)H^\perp = \Point{({A^{-1}}^T, \alpha)v}^\perp$, see Proposition \ref{prop:actionDual}, we have to keep in mind that the action is differently defined on them.

Altogether, we are going to develop a canonization algorithm for the action of $\Gsl \rtimes S_{\mathfrak{P}_0}$ on
$\prod_{i=1}^{m} (\Mat{k}{s_i})^{n_i} \times (\F_q^{k})^h$ where

\begin{equation*}
\Gsl := \left( \GL{k}(q) \times  \prod_{i=1}^{m} \GL{s_i}(q)^{n_i}  \times {\F_q^\ast}^h\right) \rtimes \Aut(\F_q)
\end{equation*}
and $\mathfrak{P}_0$ is the partition of $\Set{n+h}$ given by the different subsets.
The action of $\Gsl$ on
a vector $(U, V)$ is defined in the following way:

\begin{align*}
& \left(A, (B_i)_{i\in \Set{n}}, (b_j)_{j\in \Set{h}}, \tau^a
\right) \left((U_i)_{i\in\Set{n}}, (v_j)_{j\in\Set{h}}\right)
\\  =& \left(\left(A\tau^a(U_i) B_i^T \right)_{i\in\Set{n}},
\left({A^{-1}}^{T} \tau^a(v_j) b_j\right)_{j\in\Set{h}}\right)
\end{align*}

Before we are going to describe the rules for building up the
backtrack tree, in particular how to define an analogue inner
minimization procedure and suitable
$S_{\mathfrak{P}}$-homomorphisms, we shortly summarize which
information should be contained in each node:

\begin{itemize}
\item A permutation $\pi \in S_{n+h}$ and a standard partition $\mathfrak{P}$ of $\Set{n+h}$ describing the state of the backtrack tree analogously to Section \ref{subsec:part_refinement}.
\item A vector $(U^{(\mathfrak{P},\pi)}, V^{(\mathfrak{P},\pi)}) \in \left(\prod_{i=1}^{n} \Mat{k}{s(i), s(i)}\right) \times (\F_q^k)^h$, storing the semicanonical representative of this node.
\item A subgroup $\Inn^{(\mathfrak{P},\pi)} \leq \left(\GL{n}(q) \times \prod_{i=1}^{n} \GL{s(i)}(q) \times (\F_q^\ast)^h \right) \rtimes \Aut(\F_q)$ which stores the stabilizer under the inner group action.
\end{itemize}

Similar to Remark \ref{rem:backtrack_interpretation}, we can
interpret this backtracking as an algorithm which computes
$\Can_{S_{\mathfrak{P}_0}}\left(\Gsl(U, V)\right)$ or as the
canonization $\Can_{\Gsl \rtimes S_{\mathfrak{P}_0}}(U,V)$.
Altogether, this solves our initial canonization problem for
sequences of subsets in the projective space.

\begin{exmp}[Example \ref{ex:preprocessing} continued] \label{ex:preprocessing_cont}
We can choose

\begin{equation*}
(U^{(\mathfrak{P}_0, \id)}, V^{(\mathfrak{P}_0, \id)}) :=  \left(\begin{array}{cc:cc:cc||cccccccc|cc}
1&0&1&0&0&0  &0&0&0 &0&0 &1&1&1& 0&0\\
0&1&0&0&0&0  &0&0&0 &2&1 &0&1&2& 1&0\\
0&0&0&1&1&0  &1&1&1 &0&0 &0&0&0& 0&0\\
0&0&0&0&0&1  &0&1&2 &1&1 &0&0&0& 0&1\\
           \end{array} \right)
\end{equation*} as a semicanonical representative of the root node $(\mathfrak{P}_0, \id)$ with the initial partition
$\mathfrak{P}_0 := \{\{1,2,3\}, \{4,\ldots, 11\}, \{12,13\}\}$.

In the above representation of $(U^{(\mathfrak{P}_0, \id)},
V^{(\mathfrak{P}_0, \id)})$ we already included the partition
$\mathfrak{P}_0$: dashed lines mark the end of the matrices
$U_i^{(\mathfrak{P}_0, \id)}$, whereas solid lines mark the end of a
cell of the partition. The double horizontal line shows the change
from $U^{(\mathfrak{P}_0, \id)}$ to $V^{(\mathfrak{P}_0, \id)}$.

Before starting the backtracking, we refine $\mathfrak{P}_0$ based on the incidence graph like in the case of linear codes, see also Subsection \ref{subsubsec:incidence_graph} for a detailed description in our case:
\begin{itemize}
\item We observe that $\U_3$ is different from $\U_1, \U_2$ since it is contained in both hyperplanes in $\H_1$ whereas the other two elements are only contained in a single hyperplane.
\item The hyperplanes which contain $\U_3$ can be distinguished from all others.
\end{itemize}
This leads to the following refinement $\mathfrak{R}_0$ of the root node:

\begin{equation*}
\left(\begin{array}{cc:cc|cc||cccccc|cc|cc}
1&0&0&0& 1&0  &0&0&0  &1&1&1 &0&0& 0&0\\
0&1&0&0& 0&0  &0&0&0  &0&1&2 &2&1& 1&0\\
0&0&1&0& 0&1  &1&1&1  &0&0&0 &0&0& 0&0\\
0&0&0&1& 0&0 &0&1&2  &0&0&0 &1&1& 0&1\\
           \end{array} \right)
\end{equation*}
where we also applied the permutation $(2,3)(7,10)(8,11) \in S_{\mathfrak{P}_0}$.
\end{exmp}

In the following subsections we discuss the generalization of the
inner minimization process and the definition of suitable
$S_{\mathfrak{P}}$-homomorphisms. In Subsection
\ref{sec:Tree_revisited} at the end of this section the whole
backtracking procedure will be summarized.

\subsection{Inner Minimization}\label{subsec:InnerMin}

One main observation of \cite{Feu09} is that (in the computation of a canonical form of a linear code) the $\GL{k}(q)$-component of the stabilizer $\Inn^{(\mathfrak{P},\pi)}$ could be easily stored by a pair $(t, \mathfrak{p})$ where $0\leq t \leq k$ and $\mathfrak{p}$ is a partition of $\Set{t}$, i.e. it is equal to

\begin{flalign*}
 \GL{k}^{(t,\mathfrak{p})}(q) := \left\{\left(\left. \begin{array}{cc}
D & B_1 \\ 0 & B_2  \end{array}\right) \right| \begin{array}{l} D
\in \Mat{t}{t}\textnormal{\ diagonal matrix and constant on all } P
\in \mathfrak{p},
\\ B_1\in \Mat{t}{(k-t)}, B_2\in \GL{k-t}(q) \end{array}  \right\}.
\end{flalign*}

In the case that $\mathfrak{p}$ is the discrete partition we simply denote this group by $\GL{k}^{(t)}(q) := \GL{k}^{(t,\{\{1\}, \ldots, \{t\}\})}(q)$. We will make a similar observation in this work, but since the result is achieved by the action on the normal vectors of the hyperplanes we have to use the transposed group

\begin{equation*}
{\GL{k}^{(t,\mathfrak{p})}(q)}^T := \{ A^T \mid A \in \GL{k}^{(t,\mathfrak{p})}(q) \}
\end{equation*}
 instead. This group has furthermore the following nice property:

\begin{prop}\label{prop:row_fix}
The multiplication of some matrix $A \in {\GL{k}^{(t,\mathfrak{p})}(q)}^T$
from the left stabilizes the first $t$ rows of a matrix $U \in \Mat{k}{s}$ up to scalars.
\end{prop}

In the following, let $s(j)$ denote the number of columns of the matrix $U_j$, $j \in \Set{n}$.We describe the inner minimization process which is always applied after the partition $\mathfrak{P}$ has been refined to $\mathfrak{R}$ combined with the application of some permutation $\sigma \in S_{\mathfrak{P}}$, i.e. after each individualization step and any successful refinement. The semicanonical representative of $(\mathfrak{R}, \sigma\pi)$ is the one which is derived from the following sequence of applications of the Homomorphism Principle:
\begin{itemize}
\item Let $\Pi_{\FixedSeq{\mathfrak{R_\H}}}(U, V) := V_{\ast, \FixedSeq{\mathfrak{R_\H}}}$ and

\begin{align*}
\varphi_{\FixedSeq{\mathfrak{R_\H}}}:& \Inn^{(\mathfrak{P}, \pi)} &\rightarrow& \left(\GL{k}(q) \times {\F_q^\ast}^{\vert \Fixed{\mathfrak{R_\H}} \vert} \right) \rtimes \Aut(\F_q) \\
&(A, B, b, \tau^a) &\mapsto& (A, {(b_j)}_{j \in \FixedSeq{\mathfrak{R_\H}}},\tau^a)
\end{align*}
The pair $(\Pi_{\FixedSeq{\mathfrak{R_\H}}}, \varphi_{\FixedSeq{\mathfrak{R_\H}}})$ defines a homomorphism of group actions.
\item We will later prove that the $\GL{k}(q)$-component of the stabilizer in the previous step is a subgroup of ${\GL{k}^{(t)}(q)}^T$. In particular, we can use the parameter $t$ in the following definition:\\
Let $\Pi^{(t)}_{\FixedSeq{\mathfrak{R_\C}}}(U, V) := \left( {(U_i)}_{\Set{t}, \ast} \right)_{i\in \FixedSeq{\mathfrak{R_\C}}}$ and

\begin{align*}
\varphi_{\FixedSeq{\mathfrak{R_\C}}}^{(t)}:& \Inn^{(\mathfrak{P}, \pi)} &\rightarrow& \left(\GL{t}(q) \times \prod_{i \in \FixedSeq{\mathfrak{R_\C}}} \GL{s(i)}(q)\right) \rtimes \Aut(\F_q) \\
&(A, B, b, \tau^a) &\mapsto& (A_{\Set{t}, \Set{t}}, {(B_i)}_{i \in {\FixedSeq{\mathfrak{R_\C}}}},\tau^a)
\end{align*}
This defines a homomorphism of group actions.
\item We also apply a third homomorphism of group actions which restricts the components $\GL{s(i)}(q)$ for $i \in \Set{n} \setminus \Fixed{\mathfrak{R}_\C}$. It is only called in special cases and therefore we will state this at some more appropriate place, see Subsection \ref{par:increased_rank_C}.
\end{itemize}

\subsubsection{The structure of the inner stabilizer and semicanonical representatives}\label{subsubsec:Inn_structure}
We start with a description of the inner stabilizer
$\Inn^{(\mathfrak{P}, \pi)}$ which is computed in each step: For
simplicity we just write $\Inn$ and $(U, V)$ in the following. The
description of the computation of a semicanonical representative is
given afterward. It is based on a recursive method and the full
details are stated in Subsection \ref{subsubsec:InnerMinProc}.

After the inner minimization process, the group $\Inn$ can be
expressed by the parameters
\begin{itemize}
\item $(t, \mathfrak{p})$ -- describing the multiplication from the left
\item $(t_i)_{i\in\Set{n}}$ -- describing the multiplication from the right for each sequence element
\item $e$ -- describing the subgroup of field automorphisms
\end{itemize}
in the following way: The group $\Inn$ is the subgroup of

\begin{equation*}
\left( {\GL{k}^{(t, \mathfrak{p})}(q)}^T \times  \prod_{i=1}^{n} \GL{s(i)}^{(t_i)}(q) \times (\F_q^\ast)^h \right) \rtimes \langle \tau^e \rangle
\end{equation*}
containing all elements
$\left(\left(\begin{smallmatrix} D & 0\\ A_1 & A_2 \end{smallmatrix}\right),
{\left(\left(\begin{smallmatrix} E_i & F_i \\ 0 & G_i \end{smallmatrix}\right)\right)}_{i\in\Set{n}}, b, \tau^a\right)$ with the following properties:

\begin{align}
{\left(\begin{smallmatrix} D & 0\\ A_1 & A_2 \end{smallmatrix}\right)^T}^{-1} v_j  b_j = v_j,&\ \forall\ j \in \Fixed{\mathfrak{P}_{\H}} \label{eq:v_fix} \\
D {\left(U_i\right)}_{\Set{t}, \Set{t_i}}  E_i^T = {\left(U_i\right)}_{\Set{t}, \Set{t_i}},&\ \forall\ i \in \Fixed{\mathfrak{P}_{\C}} \label{eq:U_fix}
\end{align}

The integer $t$ is well defined by the rank of
${V}_{\ast,\FixedSeq{\mathfrak{P}_{\H}}}$. Furthermore, the inner
minimization ensures that this matrix is in reduced row echelon
form. Similarly, the integer $t_i$ is well-defined by the rank of
the submatrix ${(U_i)}_{\Set{t},\ast}$ consisting of the first $t$
rows of $U_i$. The inner minimization produces a special structure
of these matrices, i.e.: ${(U_i)}_{\Set{t}, \ast} = \left(
{(U_i)}_{\Set{t}, \Set{t_i}} , 0\right)$ where ${(U_i)}_{\Set{t},
\Set{t_i}} \in \Mat{t}{t_i}$ is in reduced column echelon form up to
scalars.

The partition $\mathfrak{p}$ is equal to the finest partition whose cells contain the
supports of all vectors $v_j, j \in \Fixed{\mathfrak{P}_{\H}}$ and the supports of the columns of ${(U_i)}_{\Set{t}, \Set{t_i}}, i \in \Fixed{\mathfrak{P}_{\C}}$. The exponent $e$ is equal to the least positive power of the Frobenius automorphism which fixes all entries of $v_j, j \in \Fixed{\mathfrak{P}_{\H}}$ and ${(U_i)}_{\Set{t}, \Set{t_i}}, i \in \Fixed{\mathfrak{P}_{\C}}$.

\begin{cor}\label{cor:ext_inner}
For $A \in {\GL{k}^{(t, \mathfrak{p})}(q)}^T$ there exists a group element

\begin{equation*}
(A, B(A), b(A), \tau^0) \in \Inn.
\end{equation*} Furthermore,
$(I_k, I_{s(1)},\ldots, I_{s(n)}, 1^h, \tau^e)$ is an element of $\Inn$.
\end{cor}

\subsubsection{Inner Minimization Process}\label{subsubsec:InnerMinProc}
In the following, we are going to describe the inner minimization process in detail. Let $(\mathfrak{P}, \pi)$ be the partition of the predecessor node and $(\sigma\pi, \mathfrak{R})$ the actual node. For simplicity we suppose that the permutation $\sigma$ was already applied to the sequence $(U^{(\mathfrak{P}, \pi)}, V^{(\mathfrak{P}, \pi)})$ and that we have to compute
$(U^{(\mathfrak{R}, \sigma\pi)}, V^{(\mathfrak{R}, \sigma\pi)})$.

In the case that $\Fixed{\mathfrak{P}} = \Fixed{\mathfrak{R}}$ there
is nothing to do. Otherwise, we successively modify the elements of
the sequence corresponding to the indices $i \in
\Fixed{\mathfrak{R}}$ starting from those components corresponding
to the hyperplanes. Additionally, we have to give rules how to
change $\Inn$ in each step in order to guarantee that the $i$-th
entry of the semicanonical representative does fulfill Equations
(\ref{eq:v_fix}) and (\ref{eq:U_fix}).

The procedures of the next two paragraphs are converted from the
algorithm for linear codes, see \cite[Algorithm 1]{Feu09}.

\paragraph{Increased Rank}\label{par:increased_rank}
Suppose that the normal vector $v_j, j \in \Fixed{\mathfrak{R}_\H} \setminus \Fixed{\mathfrak{P}_\H}$ contains some nonzero entry in the set $\{t+1, \ldots, k\}$. In this case we can perform some elementary row operations in order to map $v_j$ to the unit vector $e_{t+1}$, i.e. there is some matrix $A = \left(\begin{smallmatrix} I_t & A_1 \\ 0 & A_2                                                                                                                                                                                                                                                                                                                                                                                                                                                                                                                                                                                                            \end{smallmatrix}\right)
 \in \GL{k}^{(t,\mathfrak{p})}$ such that $A v_j = e_{t+1}$. In particular, applying the group element $({A^{-1}}^T, I_{s(1)},\ldots, I_{s(n)}, 1^h, \tau^0) \in \Inn$ leads to this result.

The new stabilizer can be described by $t+1$ and the partition $\mathfrak{p} \cup \{\{t+1\}\}$. Condition (\ref{eq:v_fix}) ensures that this vector can not be changed anymore.

\paragraph{Same Rank}\label{par:same_rank}
In the second case, the support of the newly fixed normal vector $v_j, j \in \Fixed{\mathfrak{R}_\H} \setminus \Fixed{\mathfrak{P}_\H}$ is contained in the set $\Set{t}$. In this case we can use each block
$P \in \mathfrak{p} : P \cap \operatorname{supp}(v_j) \ne \emptyset$ in order to map the nonzero entry $(v_j)_i$ with $i := \max\left(P\cap \operatorname{supp}(v_j)\right)$ onto $1_{\F_q}$. This is done using the simultaneous multiplication of all rows indexed by $i'\in P$ with $(v_j)_i^{-1}$ which corresponds to the multiplication by the matrix
$A := \left(\begin{smallmatrix}
D & 0 \\ 0 & I_{k-t}
 \end{smallmatrix}\right) \in \GL{k}^{(t,\mathfrak{p})}(q)$ with $D_{i',i'} = (v_j)_i^{-1}$ for $i' \in P$.
Corollary \ref{cor:ext_inner} ensures that we find a group element having the necessary $\GL{k}(q)$-component ${A^T}^{-1}$. Furthermore, we can choose such a group element such that $b({A^T}^{-1})_j=1$.

Finally, we may use the remaining field automorphisms $\langle (I_k, I_{s(1)}, \ldots, I_{s(n)}, 1^h, \tau^e) \rangle \leq \Inn$ to minimize the remaining nonzero, non-identity entries of the vector $v_j$ starting from the highest index.

In the following, we are not allowed to multiply two different
nonzero elements of this vector by different units since we are only
able to revert these multiplications by the multiplication of the
whole column with the same element $b_j$. Hence, the new stabilizer
can be described by $t$ and $\mathfrak{p}' := \{P \in \mathfrak{p}
\mid P \cap \operatorname{supp}(v_j) = \emptyset\} \cup \{
\bigcup_{P \in \mathfrak{p} : P \cap \operatorname{supp}(v_j) \neq
\emptyset} P \}$. The stabilizer under the field automorphisms can
be expressed by the smallest multiple $e'$ of $e$ such that
$\tau^{e'}(v_j) = v_j$.

\begin{exmp}[Example \ref{ex:preprocessing_cont} continued] \label{ex:preprocessing_cont2}
Suppose that the target cell selection told us to split the last
cell $\{12,13\} \in \mathfrak{R}_0$ in the individualization step.
Suppose we are in the branch of the backtracking where we applied
the identity element. The refined partition $\mathfrak{P}'$ contains
two singletons $\{12\}$ and $\{13\}$. Minimizing $v_{9}$ yields the
representative

\begin{equation*}
\left(\begin{array}{cc:cc|cc||cccccc|cc|c|c}
0&1&0&0& 0&0  &0&0&0  &0&1&2 &2&1& 1&0\\ \hhline{*{6}-::}
1&0&0&0& 1&0  &0&0&0  &1&1&1 &0&0& 0&0\\
0&0&1&0& 0&1  &1&1&1  &0&0&0 &0&0& 0&0\\
0&0&0&1& 0&0 &0&1&2  &0&0&0 &1&1& 0&1\\
           \end{array} \right)
\end{equation*} and the inner stabilizer with parameters $t=1, \mathfrak{p}=\{\{1\}\}$. The horizontal line shows the parameter $t$, which decomposes the matrices $U_i$ into two submatrices. Similarly, we perform some elementary row operations on the second fixed coordinate $v_{10}$, leading to

\begin{equation*}
\left(\begin{array}{cc:cc|cc||cccccc|cc|c|c}
0&1&0&0& 0&0  &0&0&0  &0&1&2 &2&1& 1&0\\
0&0&0&1& 0&0  &0&1&2  &0&0&0 &1&1& 0&1\\  \hhline{*{6}-::}
1&0&0&0& 1&0  &0&0&0  &1&1&1 &0&0& 0&0\\
0&0&1&0& 0&1  &1&1&1  &0&0&0 &0&0& 0&0\\
           \end{array} \right)
\end{equation*} and the inner stabilizer parameter $t=2, \mathfrak{p}=\{\{1\},\{2\}\}$.

Suppose that we would also have to minimize the entry $v_7 = (2,1,0,0)^T$ in the next step under this stabilizer. Then we would be able to minimize $v_7$ by the multiplication with the matrix
$A := \left(\begin{smallmatrix}
2& 0&0\\ 0 & 1 & 0\\ 0 & 0 & I_2
 \end{smallmatrix}\right)$, i.e. we would apply a group element

\begin{equation*}
({A^{-1}}^T, B({A^{-1}}^T), b({A^{-1}}^T), \tau^0).
\end{equation*}
The stabilizer would be modified such that the $\GL{4}(3)$-component would be equal to $\GL{4}^{(2, \{1,2\})}(3)$.
\end{exmp}

For the minimization of the sequence elements $U_i$ of $U$ we
distinguish two cases. The first is, that we fixed at least one
normal vector of a hyperplane whose minimization led to an increased
parameter $t$, see Subsection \ref{par:increased_rank}. In this case
we have to update all matrices $U_i$ for $i \in \Set{n}$, i.e. we
have to perform the following procedure:

\paragraph{Increasing rank for the sequence of newly fixed normal vectors}\label{par:increased_rank_C}
In this case, we know that the $\GL{k}$-component
$\Pi_{\GL{k}}(\Inn)$ of $\Inn$ is a subgroup of ${\GL{k}^{(t)}}^T$
after finishing \ref{par:increased_rank}. Proposition
\ref{prop:row_fix} ensures that these matrices stabilize
${(U_i)}_{\Set{t},\ast}$ up to scalar multiplications of the rows.
Hence, we can use the action from the right in order to map
${(U_i)}_{\Set{t},\ast}$ onto its well-defined reduced column
echelon form
$\operatorname{RCEF}\left({(U_i)}_{\Set{t},\ast}\right)$. In fact,
it is sufficient to produce the reduced column echelon form up to
multiplications of the columns by elements in $\F_q^\ast$. This is
due to the fact that at this point there is no decision on the
ordering of the elements of those cells of $\mathfrak{R}_\C$ which
are not singletons.

This canonization corresponds to the missing third homomorphism of
group actions:

\begin{alignat*}
 \theta(U, V) &:=& \left( \left( \GL{t}^{(t)}(q) \times \GL{s(i)}^{(s(i))}(q) \right) \operatorname{RCEF}(\operatorname{Rows}_t(U_i))\right)_{i\in \Set{n}} \\
 \varphi^{(t)}_{\Set{n}} : \Inn &\rightarrow& \left( {\GL{t}(q)}^T \times \prod_{i=1}^n \GL{s(i)}(q) \right) \rtimes \Aut(\F_q)\\
  (A, B, b, \tau^a)  &\mapsto&  (A_{\Set{t},\Set{t_i}}, B, \tau^a)
\end{alignat*}

\begin{rem}\label{rem:subset}
The mapping $\Gsl(U, V) \mapsto \Can_{\left( {\GL{t}(q)}^T \times \prod_{i=1}^n \GL{s(i)}(q) \right) \rtimes \Aut(\F_q)} \left( \theta(U, V) \right)$ is an $S_{\mathfrak{P}_\C}$-homomorphism, but it is not $S_{\mathfrak{P}_\C}$-invariant. Therefore, we also could derive a refinement of $\mathfrak{P}_\C$ in the application of this minimization. Nevertheless, we included this refinement in the inner minimization process, since it allows us to reduce the inner stabilizer, too.
\end{rem}

\begin{exmp}[Example \ref{ex:preprocessing_cont2} continued] \label{ex:preprocessing_cont3}
Since the minimization of the normal vectors changed the parameter $t$ of the inner stabilizer, we have to perform this step for each $i\in \Set{3}$, resulting in $U$ equal to

\begin{equation*}
\left(\begin{array}{cc:cc|cc}
1&0&0&0& 0&0 \\
0&0&1&0& 0&0 \\ \hhline{*{6}-}
0&1&0&0& 1&0 \\
0&0&0&1& 0&1 \\
           \end{array} \right).
\end{equation*}
Since we want to preserve the $(2\times 1)$ upper left submatrices of $U_1$ and $U_2$ up to row and column multiplications, we choose the parameters $(t_1, t_2, t_3) = (1,1,0)$. This also leads to a refinement of the cell $\{1,2\}$.
\end{exmp}

Finally, in a last step we guarantee the minimality of ${(U_i)}_{\Set{t},\ast}$ for those indices $i \in \Fixed{\mathfrak{R}_\C}$. We use the methods described in the following in order to guarantee that ${(U_i)}_{\Set{t},\Set{t_i}}$ is minimal and unchanged.

\paragraph{Singletons}\label{par:singletons}
This procedure is applied to all $i \in \Fixed{\mathfrak{R}_\C}$. We
can prove that the set of matrices we are allowed to apply on this
component via a multiplication with the transpose from the right is
equal to the subgroup ${\GL{s(i)}^{(t_i, \mathfrak{p}_i)}}$ for some
partition $\mathfrak{p}_i$ of $\Set{t_i}$. The minimization in
Subsection \ref{par:increased_rank_C} guarantees that the submatrix
${(U_i)}_{\Set{t}, \ast}$ decomposes as ${(U_i)}_{\Set{t}, \ast} =
\left( {(U_i)}_{\Set{t}, \Set{t_i}} , 0\right)$ where
${(U_i)}_{\Set{t}, \Set{t_i}}$ is up to scalars a $(t \times
t_i)$-matrix in reduced column echelon form. The action of
${\GL{k}^{(t, \mathfrak{p})}}$ from the left and ${\GL{s(i)}^{(t_i,
\mathfrak{p}_i)}}$ from the right only multiplies rows or columns of
${(U_i)}_{\Set{t}, \Set{t_i}}$ by nonzero entries of the field
$\F_q$. We can easily compute the smallest possible image of
${(U_i)}_{\Set{t}, \Set{t_i}}$ under this simultaneous action by
examining the operation column by column (the same algorithm as for
the normal vectors, cf. \ref{par:same_rank}). Similarly, in a
subsequent step, we use the remaining automorphisms of the field,
i.e. the group $\langle \tau^e \rangle$ for the minimization of the
nonzero entries.

\begin{exmp}
We give a bigger example over $\F_{16}$. Let $\xi \in \F_{16}$ be a primitive element and suppose we have to minimize

\begin{equation*}
U_i := \left(\begin{array}{ccc}
\xi^8 & 0 & 0\\
0 & \xi^{10} & 0\\
0 & \xi^8 & 0 \\
\xi^{12} & \xi^{4} &0 \\ \hline
\xi & \xi^7 & \xi^2
       \end{array}\right)
\end{equation*}
under the action of the inner stabilizer, whose $\GammaL{5}(16)$-component is defined by $t=4, \mathfrak{p}=\{\{1,2\},\{3\},\{4\}\}$ and $e=1$. The index $i \in [n]$ should be a newly produced singleton of the actual partition $\mathfrak{P}$ and we already produced a reduced column echelon form (up to scalars) by the methods in the previous paragraph.

First of all we must minimize the first column leading to $\overline{U}_i$ using the diagonal matrix with diagonal entries $(\xi^7, \xi^7, 1, \xi^{3})$.

\begin{align*}
\overline{U}_i := \left(\begin{array}{ccc}
1 & 0 & 0\\
0 & \xi^{2} & 0\\
0 & \xi^8 & 0 \\
1 & \xi^{7} &0 \\ \hline
\xi & \xi^7 & \xi^2
       \end{array}\right) & \quad
\widetilde{U}_i := \left(\begin{array}{ccc}
1 & 0 & 0\\
0 & \xi^{10} & 0\\
0 & 1 & 0 \\
1 & 1 &0 \\ \hline
\xi & \xi^7 & \xi^2
       \end{array}\right) & \quad
U_i' := \left(\begin{array}{ccc}
1 & 0 & 0\\
0 & \xi^{5} & 0\\
0 & 1 & 0 \\
1 & 1 &0 \\ \hline
\xi^2 & \xi^{14} & \xi^{4}
       \end{array}\right)
\end{align*}
The Frobenius automorphism stabilizes all entries of the newly produced column.
In the following, we are only allowed to apply diagonal matrices which are constant on the cells $\{1,2,4\}$ and $\{3\}$. In the next step, we minimize the second column using those row multiplications. In order to map the entry $\xi^7$ to $1$ we use the simultaneous multiplication by $(\xi^8, \xi^8, 1, \xi^8)$. Note that we can revert this multiplication on the first column via the multiplication of the whole column with $\xi^7$. Furthermore, we multiply the third row by $\xi^7$, leading to $\widetilde{U}_i$. Finally, we use the Frobenius automorphism to minimize the only nonzero and non-identity entry $\xi^{10}$ to $\xi^5$, see $U_i'$.
The $\GammaL{5}(16)$-component of the inner stabilizer is defined by $t'=4, \mathfrak{p}'=\{1,2,3,4\}$ and $e'=2$.
\end{exmp}

\subsection{Refinements}\label{subsec:refinements}
One of the most crucial tasks in the algorithm is the pruning of subtrees. We have already mentioned how the group of known automorphisms could be used for this task. Furthermore, Subsection \ref{subsec:InnerMin} allows us to prune subtrees based on the semicanonical representative of the node. Since those mappings heavily depend on $\Inn^{(\mathfrak{P},\pi)}$, they usually will have  poor performance on the first levels of the search tree. The same holds for the homomorphism we are going to introduce in the next subsection.

\subsubsection{Inner Minimization Refine}\label{subsub:innerMin_refine}
Similar to Remark \ref{rem:subset}, the mapping

\begin{align*}
 \theta^{\operatorname{Subset}} : \Gsl \pi (U,V) \mapsto \left( \operatorname{Subset}_{\Set{t}}(\operatorname{supp}(v^{(\mathfrak{P},\pi)}_j))\right)_{j\in\Set{h}},
\end{align*}
where $\operatorname{Subset}_{\Set{t}}( X ) :=
\begin{cases}
1,& \textnormal{if } X \subseteq \Set{t}\\
0,& \textnormal{else}
\end{cases}$ for $X \subseteq \Set{k}$, defines an $S_{\mathfrak{P}_{\H}}$-homomorphism.

Furthermore, we can predict the result of the inner minimization for the unfixed normal vectors and use this is as an $S_{\mathfrak{P}_{\H}}$-homomorphism as well:

\begin{equation*}
 \theta^{\operatorname{min}, \H} : \Gsl \pi (U,V) \mapsto \left( \min_{ (A,\tau^a) \in \GL{k}^{(t, \mathfrak{p})}(q)\rtimes \langle\tau^e\rangle } (A \tau^a(v^{(\mathfrak{P},\pi)}_j) \right)_{j\in\Set{h}}
\end{equation*}
This function is easily computable. In the case that $\theta^{\operatorname{Subset}}((U,V))_j = 0$ we know that $\theta^{\operatorname{min}, \H}(U, V)_j = e_{t+1}$. Otherwise, we can use the methods described in Subsection \ref{par:same_rank} in order to compute the smallest possible representative.

We can make similar computations on the positions $i\in \Set{n}$, i.e. we can define the $S_{\mathfrak{P}_{\C}}$-homomorphism $\theta^{\operatorname{min}, \C}$ with

\begin{equation*}
{\left(\theta^{\operatorname{min}, \C}(\Gsl \pi (U,V)) \right)}_i := \min_{\renewcommand{\arraystretch}{0.7} \begin{array}{c}
\scriptstyle{(D, E,\tau^a)} \in \\
\scriptstyle{(\GL{t}^{(t,\mathfrak{p})}(q) \times \GL{t_i}^{(t_i)}(q)) \rtimes \langle\tau^e\rangle  }                                                          \end{array}} D \tau^a\left( {( U^{(\mathfrak{P},\pi)}_i)}_{\Set{t},\Set{t_i}} \right) E^T
\end{equation*}
Again, this function is easily computable using the methods described in Subsection \ref{par:singletons}.

\subsubsection{Colored Incidence Graph}\label{subsubsec:incidence_graph}
Associated with the semicanonical representative $(U^{(\mathfrak{P},\pi)}, V^{(\mathfrak{P},\pi)})$ of a node is the bipartite Graph
$G$ with vertex set $\Set{n+h}$ and  edges

\begin{equation*}
\left\{ \{i,j\} \mid i \in \Set{n}, j \in (\Set{n+h} \setminus \Set{n}) : {v^{(\mathfrak{P},\pi)}_{j-n}}^T U_i^{(\mathfrak{P},\pi)} = 0 \right\} .
\end{equation*}

Using the partition $\mathfrak{P}$ we may cell-wise count the neighbors of a vertex $u \in \Set{n+h}$ which defines an $S_{\mathfrak{P}}$-homomorphism.

Finally, we have the possibility to color the edges of this graph as well. Therefore, we investigate the result of
${v^{(\mathfrak{P},\pi)}_j}^T U_i^{(\mathfrak{P},\pi)} \neq 0$ under the action of the inner stabilizer.
 For some arbitrary $A \in {\GL{k}^{(t,\mathfrak{p})}(q)}^T$,
$B:=\begin{pmatrix}
 E & B_1 \\  0 & B_2
 \end{pmatrix} \in \GL{s(i)}^{(t_i)}(q), b_j \in \F_q^\ast$ and $\tau^a \in \langle \tau^e \rangle$ we have

\begin{equation*}
 \left( {A^T}^{-1} \tau^a( v^{(\mathfrak{P},\pi)}_j) b_j\right)^T \left( A \tau^a(U_i^{(\mathfrak{P},\pi)}) \begin{pmatrix}
  E & B_1 \\ 0  & B_2
 \end{pmatrix}^T \right)=
b_j \tau^a\left(  {v^{(\mathfrak{P},\pi)}_j}^T U_i^{(\mathfrak{P},\pi)} \right)\begin{pmatrix}
  E & B_1 \\ 0  & B_2
 \end{pmatrix}^T
\end{equation*}
Now, substitute ${v^{(\mathfrak{P},\pi)}_j}^T U_i^{(\mathfrak{P},\pi)} = (w_1, w_2)$ with $w_1 \in \F_q^{t_i}$ and $w_2 \in \F_q^{s(i)-t_i}$. The action of $\Inn^{(\mathfrak{P},\pi)}$ changes the result of this product as given in the equation above. Hence, we can distinguish the edges (introduce colors) based on the orbits of $(w_1, w_2)$ under the group action of
$\left(\GL{s(i)}^{(t_i)} \times \F_q^{\ast}\right) \rtimes \langle \tau^e \rangle$.

Canonical representatives of these orbits could be easily computed using the following observation:
In the case that
$w_2 \neq 0$, we are able to find some matrix $B$ to map the vector $(w_1, w_2)$ onto the $(t_i+1)$-th unit vector.
In the case that $w_2$ is equal to $0$, we observe that $E$ is a diagonal matrix, hence the support of the vector $w_1$ is fixed by the application of this matrix.

Again, cell-wise counting of neighbors distinguished by the coloring of the edges defines an $S_{\mathfrak{P}}$-homomorphism.
The condition (\ref{eq:U_fix}) for the positions $i \in \Fixed{\mathfrak{P}_\C}$ furthermore restricts the diagonal matrices $E$ even more. In this case, it is even possible to give a refined coloring on the edges which allows the definition of a stronger $S_{\mathfrak{P}}$-homomorphism.

\begin{exmp}
Suppose that the partition $\mathfrak{P}$ of the actual node contains the cells $\{1,2\}$ and $\{1+n, 2+n\}$. Furthermore, the inner stabilizer is defined by $t=3, \mathfrak{p}=\{\{1\}, \{2,3\}\}$ and $e=2$. The example should be over $\F_4$ with $k=6$ and

\begin{align*}
(U_1, U_2, v_1, v_2) = \left(
\begin{array}{cccc:cccc||cc}
1&0&0&0& 1&0&0&0& 1&1 \\
0&1&0&0& 0&1&0&0& 1&0 \\
0&1&0&0& 0&1&0&0& 0&\xi \\ \hline
0&0&1&0&   1&\xi&1&\xi& 1&1 \\
0&\xi&0&1& 0&0&1&\xi^2& 1&0 \\
0&0&1&1&   1&\xi&1&\xi& 1&1
\end{array} \right)
\end{align*}
We first give the results of the multiplication $v_j^T U_i$:
\begin{center}
 \begin{tabular}{c|cc}
$(w^{(i,j)}_1, w^{(i,j)}_2)$  & i=1 & i=2 \\ \hline
j=1 & $(1,\xi^2,0,0)$ & $(1,1,1,\xi^2)$ \\
j=2 & $(1,\xi,0,1)$ & $(1,\xi,0,0)$
 \end{tabular}
\end{center}
There is exactly one entry in each column and each row whose
$w^{(i,j)}_2$ component is equal to zero. The support of the corresponding $w^{(i,j)}_1$ part is in both cases equal to $\{1,2\}$. Hence, the observation of the vertex- and edge-colored graph would not lead to refinement in this case.

In contrary, if the partition $\mathfrak{p}$ would be $\{\{1,2,3\}\}$ instead, we would observe that the matrices which are multiplied from the right to the elements $U_1$ and $U_2$ under the action of the inner stabilizer must be elements in $\GL{4}^{(2, \{1,2\})}(4)$. Hence, the orbits under this restricted action would be $\{(1,\xi^2,0,0), (\xi,1,0,0), (\xi^2,\xi,0,0)\}$ and $\{(1,\xi,0,0)$, $(\xi,\xi^2,0,0)$, $(\xi^2,1,0,0)\}$. Therefore, we would be able to color the edges $(1,1+n)$ and $(2,2+n)$ differently. This leads to a refinement of both cells.
\end{exmp}

\subsubsection{Iterative Refinements}
In the case that one of these refinements leads to a new singleton in the partition $\mathfrak{P}$, we use the inner minimization procedure of Section \ref{subsec:InnerMin} in order to get some smaller group $\Inn^{(\mathfrak{P},\pi)}$. The result of this minimization is compared with the candidate for the canonical form to prune the tree at an early stage.

In the case that the inner minimization leads to a smaller group $\Inn^{(\mathfrak{P},\pi)}$, we use the refinements described in Subsection \ref{subsub:innerMin_refine} immediately after the inner minimization procedure.

Since the refinement based on Subsection \ref{subsubsec:incidence_graph} is the most expensive, we try to avoid its application as long as possible. In the case that all other refinements fail and that we have updated $\Inn^{(\mathfrak{P},\pi)}$ by some smaller group or that we have replaced the partition $\mathfrak{P}$ by a finer partition since the last call of this function, the incidence graph may provide a refinement.

Rather than computing the $S_\mathfrak{P}$-homomorphism for all indices $i \in \Set{n+h}$ at once, we compute the result iteratively for each pair $P \in \mathfrak{P}_\H, Q \in \mathfrak{P}_\C$ of cells and call the refinement after each step. The ordering of the cells in this regard should be determined based on the information of all previous steps and we are not yet sure about the optimal strategy.

\subsection{The algorithm}\label{sec:Tree_revisited}

Algorithm \ref{alg:backtrack} gives a recursive description of the
backtrack tree generation and Figure \ref{img:backtracking}
visualizes this process. As we already mentioned, this tree is
traversed in a depth-first search approach. Therefore we store some
candidate for the canonical form in a global variable $(U^\Can,
V^\Can)$. This is the element which compares less than all other
leaf nodes already visited including all comparisons performed on
the paths to these nodes. At the end, this candidate is defined to
be the unique representative of the orbit. If this variable is
uninitialized ($= NIL$), the leaf node which will be visited next
becomes the candidate for the canonical form. Two further global
variables $A, T$ maintain the group of known automorphisms, which
could be used for further pruning the tree. The subroutine
\textsc{TargetCell} chooses the target cell in this step. Similarly,
\textsc{RefinementFunc} defines the next
$S_{\mathfrak{P}'}$-homomorphism which has to be applied. This
function also may return $NIL$ which indicates that the refinement
process should finish.

The function \textsc{InnerMinimization} implements the inner minimization as described in Section \ref{subsec:InnerMin}. In the case that the result of the minimization $(U', V')$ is smaller than $(U^\Can, V^\Can) \neq NIL$, this function sets the global variable
$(U^\Can, V^\Can)$ to $NIL$. On the other hand, if there is a smaller candidate for the canonical form, the function returns $is\_leq = \textbf{false}$. Otherwise, this flag is set to \textbf{true}. Similarly, \textsc{Refinement} implements the Homomorphism Principle for some given $S_{\mathfrak{P}'}$-homomorphism $f$. It updates the variables in the same way. Furthermore, it also calls the \textsc{InnerMinimization} function in the case that a further singleton appeared in $\mathfrak{P}'$.

\begin{algorithm}[tb]
  \caption[Backtrack]{\textsc{Backtrack}}
  \label{alg:backtrack}
  \begin{algorithmic}[1]
  \Require global variable $(U^\Can, V^\Can)$ -- candidate for the canonical form
  \Require global variable $\pi^\Can$ -- the permutation leading to the candidate
  \Require global variable $Aut \leq S_{n+h}$ -- the group of known automorphisms
  \Require global variable $T \subseteq S_{n+h}$ -- a left transversal of $Aut$ in $S_{n+h}$
  \Require $(\mathfrak{P}, \pi, U,V)$ a node of the backtrack tree
  \Ensure individualization-refinement step on $(\mathfrak{P}, \pi, U,V)$
\Procedure{Backtrack}{$\mathfrak{P}, \pi, U,V, \Inn$}
  \If{$\mathfrak{P}$ is discrete}
    \If{$(U^\Can, V^\Can) = NIL$}
      \State $(U^\Can, V^\Can,\pi^\Can) \leftarrow (U,V, \pi)$ \Comment{a new candidate}
    \Else
      \State $Aut \leftarrow \left\langle Aut, \pi^{-1}\pi^\Can \right\rangle$ \Comment{a new automorphism}
      \State $T \leftarrow$ left transversal of $Aut$ in $S_{n+h}$
    \EndIf
    \State \Return
  \EndIf

  \State $P \leftarrow \textsc{TargetCell}( \mathfrak{P}, \pi, \Inn )$, $m \leftarrow \operatorname{min}(P)$ \Comment{target cell selection}
  \State $\mathfrak{P}' \leftarrow (\mathfrak{P} \setminus P) \cup \{\{m\}, P \setminus \{m\}\}$
  \For{$j \in P$}\Comment{individualization}
    \State $t \leftarrow (m, j)$, $\pi' \leftarrow t\pi$, $(U',V' ) \leftarrow t(U, V)$
    \State $\left( is\_leq,(\mathfrak{P}', \pi', U',V', \Inn') \right) \leftarrow \textsc{InnerMinimization}(\mathfrak{P}', \pi', U',V', \Inn)$
    \State $f \leftarrow \textsc{RefinementFunc}(\mathfrak{P}', \Inn')$ \Comment{choose an $S_{\mathfrak{P}'}$-homomorphism}
    \While{$is\_leq$ \textbf{and} $f \neq NIL$}
    \If{$T \cap S_{\mathfrak{P}'}\pi = \emptyset$}\Comment{see \cite[Lemma 5.9]{Feu09}}
    \State \Return
    \EndIf
    \State $\left( is\_leq,(\mathfrak{P}', \pi', U',V', \Inn')\right)\leftarrow \textsc{Refinement}(f, \mathfrak{P}', \pi', U',V', \Inn')$
    \State $f \leftarrow \textsc{RefinementFunc}(\mathfrak{P}', \Inn')$
    \EndWhile
    \If{\textbf{not} $is\_leq$ }
      \State \Return \Comment{the actual candidate $(U^\Can, V^\Can)$ is smaller}
    \EndIf
    \State\textsc{Backtrack}($\mathfrak{P}', \pi', U',V', \Inn'$)
  \EndFor
\EndProcedure
\end{algorithmic}
\end{algorithm}

\begin{figure}[bt]
\includegraphics[width=\linewidth]{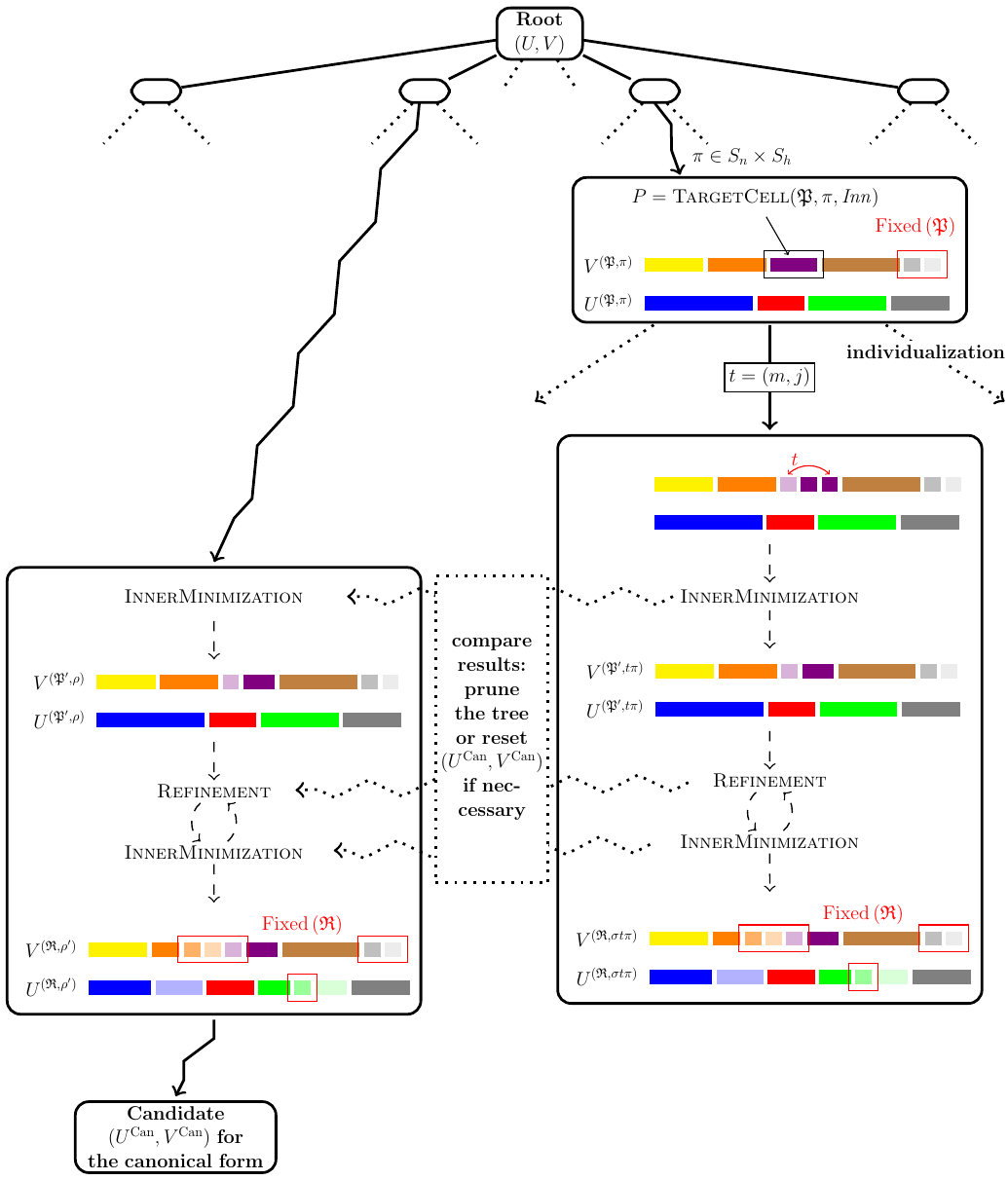}
\caption{Backtrack tree generation}\label{img:backtracking}
\end{figure}

During the backtracking it is not necessary to maintain a group
element $(A, B, b, \tau^a) \in \Gsl$, which maps the root node to
the semicanonical representative of the actual node. The permutation
$\pi \in S_{n+h}$ and the corresponding path to the leaf
$(\mathfrak{D},\pi)$ define an element $(A^{(\pi)}, B^{(\pi)},
b^{(\pi)}, \tau^{a^{(\pi)}}) \in \Gsl$ which maps the initial
sequence $(U, V)$ of the root node to the semicanonical
representative of this leaf. The element $(A^{(\pi)}, B^{(\pi)},
b^{(\pi)}, \tau^{a^{(\pi)}})$ is well defined up to the
multiplication by $\Inn^{(\mathfrak{D},\pi)}$ from the left and it
is only computed for some few leaves. Furthermore, since we are only
interested in a canonization map for $\C$, we may restrict this
computation to its $\GammaL{k}(q)$-component $\left(A^{(\pi)},
\tau^{a^{(\pi)}}\right)$.

Let $\pi^\Can \in S_{n+h}$ be the permutation leading to the canonical form
and let $\sigma_1, \ldots, \sigma_z \in S_{n+h}$ define generators of the automorphism group $Aut$ used for pruning the search tree.
The canonical form of $\Can_{\GammaL{k}(q)}(\C)$ is defined to be the sequence of subsets of subspaces given by the column spaces of $U_i^\Can$. A transporter element $\TR_{\GammaL{k}(q)}(\C)$ is given by $\left(A^{(\pi^\Can)}, \tau^{a^{(\pi^\Can)}}\right)$.

\begin{prop} With the help of the elements $\left(A^{(\pi^\Can\sigma_i)}, \tau^{a^{(\pi^\Can\sigma_i)}}\right)$, $i\in \Set{z}$, the group $\Inn^{(\mathfrak{D},\pi^\Can)}$
and $\left(A^{(\pi^\Can)}, \tau^{a^{(\pi^\Can)}}\right)$ we are able to compute generators of the automorphism group $\Aut(\C) \leq \GammaL{k}(q)$: $\Aut(\C)$ is generated by
\begin{flalign*}
\Pi_{\GammaL{k}}\left(\Inn^{(\mathfrak{D},\id)}\right) =
\left(A^{(\pi^\Can)}, \tau^{a^{(\pi^\Can)}}\right)^{-1} \Pi_{\GammaL{k}(q)}\left(\Inn^{(\mathfrak{D},\pi^\Can)}\right) \left(A^{(\pi^\Can)} , \tau^{a^{(\pi^\Can)}}\right)  \\
\textnormal{and }
\left\{\left(A^{(\pi^\Can)}, \tau^{a^{(\pi^\Can)}}\right)^{-1}\left(A^{(\pi^\Can\sigma_i)}, \tau^{a^{(\pi^\Can\sigma_i)}}\right) \mid i\in \Set{z} \right\}
\end{flalign*}
\end{prop}

%

Note that $(\mathfrak{D},\pi^\Can\sigma_i)$ may not appear as a node of the pruned search tree because of the pruning based on the group of known automorphisms. In this case, we still know in which order we have to apply the methods from Section \ref{subsec:InnerMin} to compute its semicanonical representative, since it defines the same canonical form and hence the inner minimization has to follow the same rules than the computation of $\left(A^{(\pi^\Can)} , \tau^{a^{(\pi^\Can)}}\right)$.

\begin{thm}
The mapping $\C \mapsto \left( \Can_{\GammaL{k}(q)}(\C), \TR_{\GammaL{k}(q)}(\C), \Aut(\C) \right)$, where $\Can_{\GammaL{k}(q)}$ and
$\TR_{\GammaL{k}(q)}$ are defined as above, solves the canonization problem and Algorithm \ref{alg:backtrack} is a practical algorithm to compute the data.
\end{thm}

\section{Applications}\label{sec:Applications}

In \cite{Feu09, FeuZ4} we gave running times for the computation of
the automorphism groups of a family of almost perfect nonlinear
(APN-) function $f^{(d)}: \F_2^d \rightarrow \F_2^d, x \mapsto x^3$.
These computations were done using a reformulation as a linear code
$C_f^{(d)} \subseteq \F_2^{2^d}$ of dimension $2d+1$.

\begin{defn}
A set $\C$ of $d$-dimensional subspaces in $\PS{q}{k+1}$ with $\vert \C \vert = 1 + (q^{d+1}-1)/(q-1)$ is called a
$(d-1)$-dimensional dual hyperoval if
the intersection of any two distinct elements of $\C$ is a point and
any three have an empty intersection.
\end{defn}

\cite{Yoshiara} gives a construction of $(d-1)$-dimensional dual
hyperovals $\C_f^{(d)}$ in $\PS{2}{2d}$ using quadratic
APN-functions $f^{(d)} : \F_2^{d} \rightarrow \F_2^{d}$. Again we
use the quadratic function $x \mapsto x^3$ to produce a family of
subsets $\C_f^{(d)}$ to test our algorithm. By \cite{EdelDempwolff}
we know that the automorphism group of $\C_f^{(d)}$ and $C_f^{(d)}$
are identical for $d \geq 4$. This allows us on the one hand to test
the algorithm for correctness and on the other to compare its
performance with the algorithm for linear codes. Table
\ref{tab:RunningTimes} shows the running times for different $d$ on
a single core of a 2.4 GHz Intel Quad 2 processor.

\begin{table}
\begin{center}
\begin{tabular}{ccc|c:c|cc}
$k=2d$ & $s=d$ & $n=2^d$ & $h$ & $\vert \Aut(\C_f^{(d)}) \vert$ & time $\C_f^{(d)}$ & time $C_f^{(d)}$\\ \hline
6 & 3 & 8 & 28 & 1344 & 0.1 s & 0.1 s \\
8 & 4 & 16 & 20 & 5760 & 0.1 s & 0.1 s \\
10 & 5 & 32 & 496 & 4960 & 0.5 s & 0.1 s \\
12 & 6 & 64 & 336 & 24192 & 0.2 s & 0.1 s\\
14 & 7 & 128 & 8128 & 113792 & 3 s & 0.3 s \\
16 & 8 & 256 & 5440 & 522240 & 2.5 s & 0.3 s\\
18 & 9 & 512 & 130816 & 2354688 & 4 min & 45 s \\
20 & 10 & 1024 & 87296 & 10475520  & 2 min & 6 s \\
22 & 11 & 2048 & 2096128 & 46114816 & 8 h & 4 h \\
24 & 12 & 4096 & 1397760 & 201277440 & 8 h & 6 min
 \end{tabular}
\caption{Running times for $\C_f^{(d)}$ compared to $C_f^{(d)}$ for
$f(x)=x^3$}
\end{center}\label{tab:RunningTimes}
\end{table}

\section{Conclusion}\label{sec:Conclusion}
This works presents a practical algorithm which solves the
canonization problem for sequences of subsets of $\PS{q}{k}$. From
the reduction to the graph isomorphism problem, we know that we
could not expect to give an algorithm that runs in polynomial time.

The algorithm itself relies on many heuristics, for instance the
choice of the target cell, the choice of the homomorphism of group
actions which has to be applied next or when to stop the refinements
since they are more expensive than performing an individualization
step. These problems are well-known from the canonization of graphs
where possible modifications of the basic algorithm are discussed in
several papers.

Similarly, the improvement of these heuristics is still part of our
current research and we are not yet sure about an optimal strategy.
This is also the reason of not giving the full implementation
details in this regard.

\begin{ack}
The author thanks Anna-Lena Trautmann and Axel Kohnert for many helpful discussions and comments. Furthermore, he would like to thank Yves Edel for the hint on the connection to additive codes and the examples for testing the implementation of the algorithm.
\end{ack}

\bibliographystyle{plain}
\bibliography{../../bibliography}


\end{document}